\documentclass[journal]{IEEEtran}
\usepackage{amsmath,amssymb}
\usepackage{subfigure}
\usepackage{graphicx,graphics,color,psfrag}
\usepackage{cite,balance}
\usepackage{caption}
\allowdisplaybreaks
\usepackage{algorithm}
\usepackage{algorithmic}
\usepackage{accents}
\usepackage{amsthm}
\usepackage{bm}
\usepackage{url}
\usepackage[english]{babel}
\usepackage{multirow}
\usepackage{enumerate}
\usepackage{cases}
\usepackage{stfloats}
\usepackage{dsfont}
\usepackage{color,soul}
\usepackage{amsfonts}
\usepackage{cite,graphicx,amsmath,amssymb}
\usepackage{subfigure}
\usepackage{fancyhdr}
\usepackage{hhline}
\usepackage{graphicx,graphics}
\usepackage{array,color}
\usepackage{mathtools}
\usepackage{amsmath}
\usepackage[T1]{fontenc}
\usepackage{float}  
\usepackage{subfigure}  

\newtheorem{definition}{\emph{\underline{Definition}}}

\newtheorem{lemma}{\emph{\underline{Lemma}}}
\newtheorem{observation}{\emph{\underline{Observation}}}
\newtheorem{corollary}{\emph{\underline{Corollary}}}

\newtheorem{proposition}{\emph{\underline{Proposition}}}

\newtheorem{remark}{\bf \emph{\underline{Remark}}}

\def\l{\left}
\def\r{\right}
\def\({\left(}
\def\){\right)}

\def\b0{{\mathbf{0}}}

\newcommand{\nn}{\nonumber}

\begin{document}
\captionsetup[figure]{name={Fig.}}

\title{\huge 
Near-Field Beam Training: Joint Angle and Range\\ Estimation with DFT Codebook} 
\author{Xun Wu, Changsheng You, \IEEEmembership{Member, IEEE}, Jiapeng Li, 
	and Yunpu Zhang
	\thanks{All authors are with the Department of Electronic and Electrical Engineering, Southern University of Science and Technology, Shenzhen 518055, China. (e-mails: wux2022@mail.sustech.edu.cn; youcs@sustech.edu.cn; lijp2023@mail.sustech.edu.cn; yunpuzhangcsu@gmail.com). \emph{(Corresponding author: Changsheng You.)} \\
	\indent This work of Changsheng You was supported in part by the National Natural Science Foundation of China under Grant 62201242, 62331023, in part by the National Key R\&D Program Youth Scientist Project under Grant 2023YFB2905100, in part by Shenzhen Science and Technology Program under Grant 20231115131633001, in part by Young Elite Scientists Sponsorship Program by CAST under Grant 2022QNRC001.	 
	 }

}\maketitle

\begin{abstract} 
Prior works on near-field beam training have mostly assumed dedicated polar-domain codebook and on-grid range estimation, which, however, may suffer long training overhead, high codebook storage requirement, and degraded estimation accuracy. 
To address these issues, we propose in this paper new and efficient beam training schemes with off-grid range estimation by using conventional discrete Fourier transform (DFT) codebook. 
Specifically, we first analyze the received beam pattern at the user when far-field beamforming vectors are used for beam scanning, and show an interesting result that this beam pattern contains useful user angle and range information. 
	Then, we propose two efficient schemes to jointly estimate the user angle and range with the DFT codebook. 
	The first scheme estimates the user angle based on a defined angular support and resolves the user range by leveraging an approximated angular support width, while the second scheme estimates the user range by minimizing a power ratio mean square error (MSE) to improve the range estimation accuracy.
	Finally, numerical simulations show that our proposed schemes greatly reduce the near-field beam training overhead and improve the range estimation accuracy as compared to various benchmark schemes.
	
\end{abstract}
\vspace*{-5pt}
\section{Introduction}

Extremely large-scale array/surface (XL-array/surface) has emerged as a promising technology for future sixth generation (6G) wireless systems to achieve ultra-high spectral efficiency and extremely high spatial resolution in high-frequency bands, hence accommodating new applications such as extended reality, holographic video, autonomous driving, etc \cite{GiordaniToward2020,CuiNear2023,WuIntelligent2021}.
By increasing the number of antennas by another order-of-magnitude, XL-arrays are expected to bring a fundamental change in the electromagnetic (EM) propagation modelling, shifting from conventional far-field radio propagation to the new \emph{near-field} channel modelling \cite{ChenBeam2023,you2023nearfield,9133142}. 

In particular, unlike the planar wavefront assumption for the far-field channel modelling, the near-field channels need to be accurately modeled based on the \emph{spherical} wavefronts, which brings both opportunities and challenges. 
Specifically, different from conventional maximum ratio transmission (MRT) based far-field beamforming that steers the beam energy at a certain angle (called far-field beam-steering), the near-field beamforming based on the spherical wavefront characteristic makes it possible to concentrate the beam energy at a certain location/region, which is termed as \emph{near-field beam-focusing} in the literature \cite{10068140,9723331,NepaNear2017}.
The new beam-focusing effect offers a new degree-of-freedom (DoF) to flexibly/dynamically control the beam energy distribution in both the angle and \emph{range} domains \cite{10098681,WeiChannel2022}.
Thanks to the larger spatial/range DoF, the rank of line-of-sight (LoS) multiple-input multiple-output (MIMO) channel in the near-field may be larger than one, which provides the potential to support multiple data streams for enhancing the multiplexing gain \cite{liu2023nearReview,10117500}.
Moreover, by judiciously controlling the beam-focusing region, XL-arrays can effectively eliminate inter-user interference to improve the communication performance \cite{luo2022beam,10147356}. 
	In addition, the near-field beam-focusing effect can also be exploited to improve the power charging efficiency of wireless power transfer (WPT) systems \cite{zhang2023joint,zhang2023swipt}, as well as the accuracy and resolution of  indoor/outdoor localization \cite{abu2021near,guidi2021radio,wang2023near,cong2023near}.
	Near-field channel brings more opportunity for localization, since it can provide more accurate user angle and distance information \cite{guidi2021radio,abu2021near}.
It is also suggested that ISAC systems can benefit more from the     near field compared with the far field \cite{wang2023near,cong2023near}.
In this paper, we propose efficient beam training methods for near-field communications and show an interesting result that the DFT codebook, conventionally used for far-field beam training, can also be utilized to find the user angle and range in the near-field beam training which achieves more accurate range estimation. 

\vspace*{-5pt}
\subsection{Related Works}
In the existing literature, near-field beamforming designs and their communication performances in various systems have been widely studied \cite{9738442,10123941,LuCommunicating2022,zhang2023mixed}. 
Specifically, the authors in \cite{9738442} considered three beamforming architectures for near-field multi-user communication system (including fully-digital architecture, hybrid digital-analog architecture, and dynamic metasurface antenna architecture), and optimized both the configuration and digital
precoding to maximize achievable sum-rate in multi-user systems.
It was shown that near-field beamforming can utilize the new range domain to effectively suppress inter-user interference, even when users are located at the same angle, but different ranges. 
Then, a new location divided multiple access (LDMA) was proposed in \cite{10123941} to increase the system capacity and enhance multi-user access ability in the near-field region. 
Moreover, the authors in \cite{LuCommunicating2022} provided a unified framework to model the discrete antenna arrays and continuous surfaces, based on which they derived a closed expression for the received signal-to-noise ratio (SNR) under optimal single-user MRT beamforming.
It was revealed that the SNR increases with the number of antennas, while suffering a diminishing return.
Furthermore, the authors in \cite{zhang2023mixed} considered a new mixed-field scenario with both near- and far-field users in wireless systems and characterized the inter-user interference using Fresnel functions.
Interestingly, they showed that when the far- and near-field users are located at similar angles, there may exist strong inter-user interference due to the new \emph{energy-spread} effect. 
Apart from near-field beamforming optimization, near-field beam training design is also of paramount importance, which provides necessary channel state information (CSI) to establish high SNR initial links before data transmission. 
Note that in high-frequency bands such as $ 30 $ GHz (mmWave), the direct channel estimation methods may not be efficient due to the channel sparsity and high path-loss in mmWave signals.  As such, the received signal power may be small or even negligible if the XL-array beamforming vector is not well aligned with the channel paths. 
Compared with conventional far-field beam training, near-field beam training is new and more challenging due to the new spherical wavefront model \cite{ZhouSperical2015}. 
As such, directly applying the DFT codebook to find the best codeword based on the maximum received signal power will result in significant performance degradation due the near-field energy-spread effect \cite{9129778,Cui2022channel}. 
This thus calls for developing new approaches to design near-field codebooks and dedicated beam training schemes. 
Specifically, for the codebook design, the authors in \cite{Cui2022channel} proposed a new polar-domain codebook, which samples the angle domain uniformly and the range domain non-uniformly.
Based on this codebook, a straightforward beam training approach is exhaustively searching over all candidate codewords in the two-dimensional (2D) polar-domain. This, however, requires a prohibitively high training overhead with the required number of training symbols being the product of numbers of antennas and sampled ranges.
To address this issue, a new two-phase near-field beam training method was proposed in \cite{Zhang2022fast}, which first estimates the user angle using the DFT codebook and then resolves the user range using the polar-domain codebook. 
Although this two-phase near-field beam training method requires a lower training overhead which is proportional to the number of antennas, it may be still unaffordable when the number of XL-array antennas is very large.
This thus motivated recent research efforts to design hierarchical near-field beam training methods to further reduce the training overhead. For example, the authors in \cite{wu2023twostage} proposed an efficient two-stage hierarchical beam training method, which firstly searches the coarse user angle based on far-field hierarchical codebook and then searches finer-grained user angle-and-range pair with a dedicated near-field hierarchical codebook. 
	In \cite{lu2022hierarchical}, the authors  first proposed a hierarchical near-field codebook,  and then devised an efficient hierarchical beam training scheme for the hybrid beamforming architecture. In addition,  the authors  in \cite{10217152} proposed to project the near-field channel into the spatial-angular and slope-intercept domains, and developed two hierarchical beam training schemes to reduce the beam training overhead.
In addition, deep learning methods were also utilized in \cite{Pan2022deep} to reduce the training overhead by utilizing two separated convolutional neural networks (CNNs) to estimate the user angle and range based on far-field wide beams.
\vspace*{-5pt}
\subsection{Motivations and Contributions} 
The existing works on near-field beam training still have some limitations. First, it is widely assumed that the DFT codebook cannot be used for estimating the user angle and range in near-field beam training, as the conventional beam training method  determines the user angle based on the largest received power at the user. 
	This, however, may not have fully exploited the received power pattern at the user; hence, it is still questionable whether the DFT codebook indeed cannot be used for estimating both the user and range in the near-field region. Second, most existing works have considered the on-grid beam training in both the angular and range domains. In particular, the polar-domain codebook usually consists of a small number of sampled ranges only (e.g., $5$ sampled ranges for the case with $256$ antennas and $30$ GHz carrier frequency). This inevitably results in low-resolution range estimation, which is undesirable for location/sensing-sensitive applications such as autonomous driving, immersive reality and telemedicine \cite{you2023nearfield}. Moreover, the polar-domain codebook requires a large storage. 

 Motivated by the above, we revisit the near-field beam training design based on the DFT codebook, and propose new and efficient schemes to jointly estimate the user angle and range.
The main contributions of this paper are summarized as follows.

\begin{itemize}
	\item First, we analyze the received beam pattern at the user when the conventional DFT codebook is used for beam sweeping. 
	\emph{Interestingly}, we show that both the user angle and range can be inferred from the received beam pattern. 
	Specifically, the user angle can be estimated based on a defined \emph{angular support} that characterizes the angle region with dominant received signal powers, while the user range can be inferred from \emph{angular support width}, which is \emph{monotonically decreasing} with the user range given fixed user angle. 
	Nevertheless, it is intractable to obtain an explicit expression for the angular support width. Thus,  we define a new \emph{surrogate angular support width}, and analytically characterize its expression with respect to (w.r.t.) the user angle and  range in a semi-closed form and discuss its main properties. 
	\item Second, we propose two efficient schemes to jointly estimate the user angle and range under the DFT codebook. 
	Specifically, the first scheme estimates the user angle by determining the middle of defined angular support, and then estimates the user range by leveraging the width of surrogate angular support. 
	This scheme, however, may suffer degraded range estimation accuracy in the low-SNR regime due to potential  inaccurate angular support width estimation. 
	To improve the range estimation accuracy, a second near-field beam training scheme is proposed, which utilizes the same user angle estimation method, while the user range is estimated by more effectively exploiting all high-SNR received powers within the surrogate angular support via minimizing sum power-ratio MSE. 
	
	\item Finally, numerical results are provided to demonstrate the effectiveness of our proposed joint angle and range estimation schemes for near-field beam training. 
	It is shown that both proposed schemes can achieve higher range estimation accuracy as compared to existing benchmark schemes based on the polar-domain codebook. 
	This is because our proposed schemes more effectively exploit the received beam pattern for achieving off-grid range estimation. 
	Moreover, in the high-SNR regime, the proposed schemes obtain higher achievable rates than the exhaustive-search 2D beam training due to the more accurate range estimation, while they suffer slightly degraded rate performance in the low-SNR regime. 
	In addition, the second scheme outperforms the first scheme in terms of achievable rate because it is more robust to noise by  exploiting all high-SNR received powers within the surrogate angular support.
\end{itemize}

It is worth noting that several near-field localization methods have also been proposed in the literature. For example,  in \cite{abu2021near}, the authors performed a Fisher information analysis and proposed a low-complexity 3D localization algorithm, which transforms the 3D problem into three 1D problems. Additionally, in \cite{guidi2021radio}, the authors investigated the possibility of inferring the position of a single antenna transmitter using a single asynchronous receiving node by extracting information from the incident spherical wavefront.
These localization methods exploit the received signal at the XL-array to estimate the location of the user, while the beam training methods studied in this paper aim to find the best user beam codeword for communication based on the received signal powers at the users.
\subsection{Organization and Notations}
\emph{Organization:}
The rest of this paper is organized as follows. 
We present the near-field system model in Section II, and introduce several existing benchmark schemes for near-field beam training in Section III.
Then, we analyze the received beam pattern at the user
in Section IV, and propose two efficient schemes to jointly estimate the user angle and range based  on conventional DFT codebook in Section V.
Numerical results are provided in Section VI, followed by conclusions made in Section VII.

\emph{Notations:}
Lower-case, upper-case boldface letters denote vectors, matrices, respectively, e.g $\textbf{a}$ and $\textbf{A}$, while $a$ and $\mathcal{A}$ represent a scale and a set.
For a vector, $(\cdot)^T$ and $(\cdot)^H$ denote its transpose and conjugate transpose, respectively. $|\cdot|$ denotes the absolute value for a number and the cardinality for a set, respectively. The main symbols used in this paper are summarized in Table~\ref{Table}.

\begin{table*}[htb]
		\renewcommand{\arraystretch}{1.5}
		  \setlength{\tabcolsep}{6pt}
		  \vspace{-10pt}
		 \centering
		 \caption{List of main symbols and their physical meanings.}
		 \label{Table}
		 \vspace*{-2pt}
	\begin{tabular}{|c|l|c|l|}
		\hline
		$N$                                                     & Number of BS antennas              & $\tilde{\mathbf{W}}_{\rm{DFT}}$                                    & DFT codebook                    \\ \hline
		$D$                                                     & Array aperture size                & $\bar{\mathbf{w}}_{\rm n,s}$                                           & Polar-domain codeword           \\ \hline
		$\lambda$                                               & Carrier wavelength                 & $\tilde{\mathbf{w}}_{n}$                                           & DFT codeword                    \\ \hline
		$\mathbf{h}^H_{\rm near}$                               & Near-field user channel            & $K$                                                                & Number of candidate user angles \\ \hline
		$h_{\rm u}$                                & LoS channel path gain                 & $S$                                                                & Number of range samples         \\ \hline
		$ \kappa $                          & Rician factor       & ${\tilde{\mathbf{w}}}(\Omega)$                                     & Far-field beamfroming vector    \\ \hline
		$\theta_{\rm u}$                                        & BS center-user spatial angle       & $\mathcal{A}_{\beta{\rm dB}}(\theta_{\rm u}, r_{\rm u})$           & Angular support                 \\ \hline
		$r_{\rm u}$                                             & BS center-user range               & $\widehat{\mathcal{A}}_{\beta{\rm dB}}(\theta_{\rm u}, r_{\rm u})$ & Surrogate angular support       \\ \hline
		$\phi_{\rm u}$                                                & Physical user angle-of-departure   & $\Gamma_{\beta {\rm dB}}(\theta_{\rm u}, r_{\rm u})$               & Angular support width           \\ \hline
		$\mathbf{b}^{H}\left(\theta_{\rm u}, r_{\rm u}\right) $ & Near-field channel steering vector & $\widehat{\Gamma}_{\beta {\rm dB}}(\theta_{\rm u}, r_{\rm u})$     & Surrogate angular support width \\ \hline
		$\bar{\mathbf{W}}_{\rm{Pol}}$                           & Polar-domain codebook              & $G(\theta_{\rm u}, r_{\rm u}, \phi)$                               & Beam power ratio                \\ \hline
	\end{tabular}
\end{table*}

\begin{figure}[t]
	\centering
	\includegraphics[width=8.5cm]{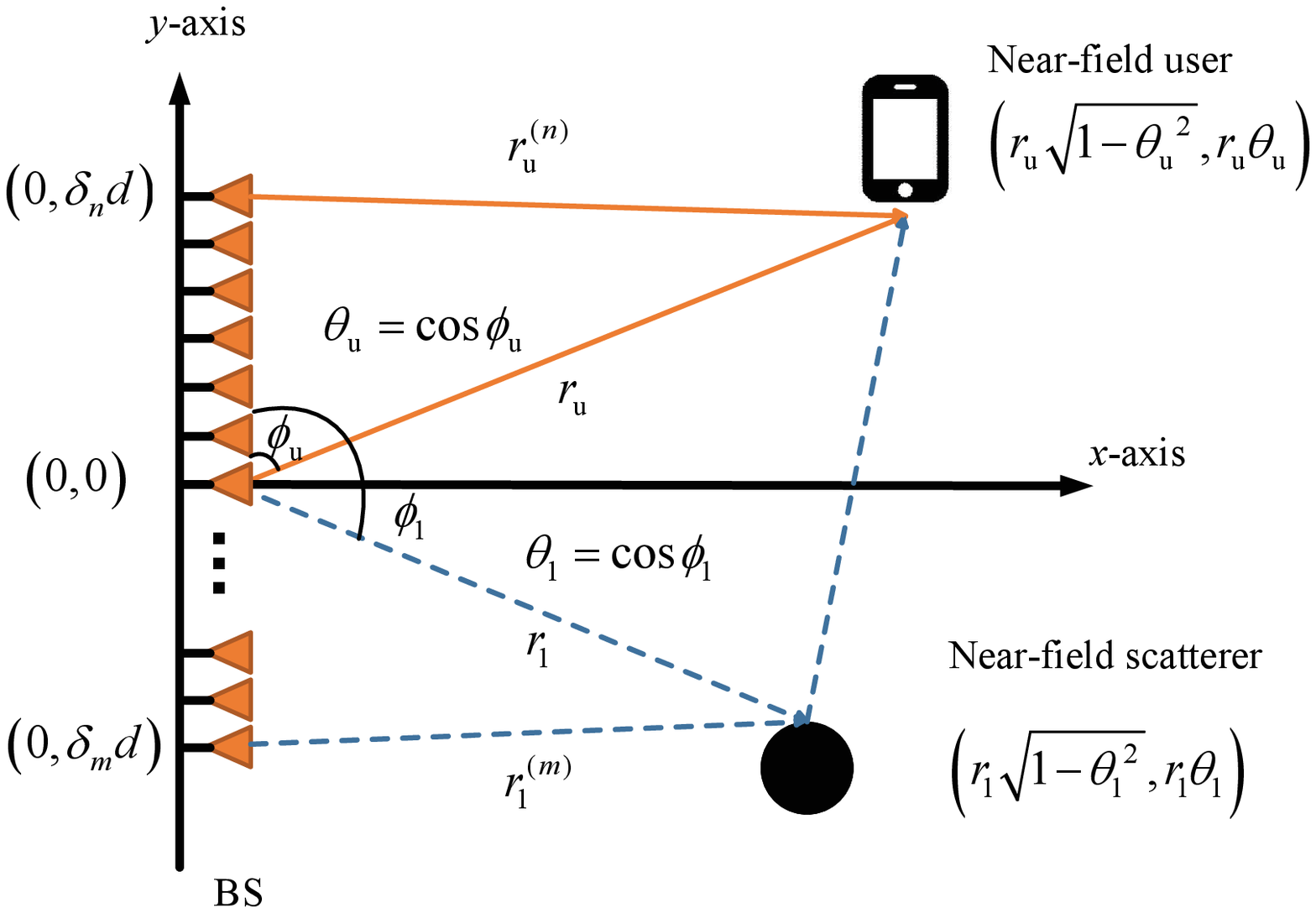}
	\caption{A narrow-band XL-array communication system.}
	\label{fig:sytem_model}
	\vspace*{-10pt}
\end{figure}
\vspace{-3pt}
\section{System model} 
We consider the downlink beam training for a narrow-band XL-array communication system as shown in Fig.~\ref{fig:sytem_model}, where the base station (BS) is equipped with $N$-antenna uniform linear array (ULA) and the user has a single-antenna.\footnote{The proposed near-field beam training schemes can be readily extended to the multi-user scenario by performing the angle-and-range estimation simultaneously for all users.} 

\underline{\bf Near-field channel model:}
We assume that the user is  located in the near-field region of the XL-array, i.e., the BS-user range is no larger  than the Rayleigh distance $Z_{\rm Rayl}=\frac{2D^2}{\lambda}$, where $D$ and $\lambda$ denote the array aperture size and  carrier wavelength, respectively. For example, when $D=1.3$ m and $f = 30$ GHz, the Rayleigh distance is approximately $327$ m.
Moreover, without loss of generality, we assume that the XL-array is placed at the $y$-axis and each antenna $n$ is located at ($0, \delta_{n}d$), where $\delta_{n}=\frac{2n-N-1}{2}$, $n\in \mathcal{N}\triangleq\{1,2,\cdots,N\}$ and $d=\frac{\lambda}{2}$ denotes the  adjacent antenna spacing. 
Then the channel from the XL-array BS to the user can be modeled as \cite{ZhangCapacity}
\begin{equation}
	\mathbf{h}^H_{\rm near} = \sqrt{N}h_{\rm u}\mathbf{b}^{H}(\theta_{\rm u},r_{\rm u}) + \sqrt{\frac{N}{L}}\sum_{\ell=1}^{L}h_{\ell} \mathbf{b}^H (\theta_{\ell},r_{\ell}),
\end{equation}
which includes one LoS path and $L$ non-LoS (NLoS) paths. Parameters $h_{\ell}$, $\theta_{\ell}$, and $r_{\ell}$ denote the complex gain, spatial angle and range of the $\ell^{th}$ path, respectively. The index $\ell = \rm u$ represents the LoS path, while $\ell \ge 1$ denotes the NLoS path. Specifically, the complex gain for the LoS path is given by $h_{\rm u}= \sqrt{\frac{\kappa}{\kappa+1}}\frac{\sqrt{\beta}}{r_{\rm u}} e^{-\frac{\jmath 2 \pi r_{\rm u}}{\lambda}}$, while the complex gain for NLoS path is $h_{\ell} \sim \mathcal{CN} (0,\sigma^2_{\ell})$, where $\sigma_{\ell} = \sqrt{\frac{1}{\kappa+1}} \frac{\sqrt{\beta}}{r_{\rm u}}$ \cite{10123941,ZhangCapacity}.  Parameter $\kappa$, $r_{\rm u}$, and $\beta$ represent the Rician factor, the BS center-user range and the reference channel gain at a range of $1$ m, respectively.
In this paper, we focus on the near-field communication scenarios in high-frequency bands such as millimeter-wave (mmWave) and even terahertz (THz) for which the NLoS channel paths have negligible power due to the severe path-loss and shadowing\cite{10130629,9241752,9528043}. 
Therefore, when the LoS path is dominant (i.e., $ \kappa $ is very large),  the BS$\to $user channel can be approximated as $\mathbf{h}^H_{\rm near} \approx  \sqrt{N}h_{\rm u} \mathbf{b}^{H}(\theta_{\rm u},r_{\rm u})$.
Based on the spherical wavefront propagation model and assuming the BS-user range larger than the \emph{uniform power range}\cite{liu2023nearReview,you2023nearfield}, the near-field channel steering vector is defined as \cite{Zhang2022fast}\footnote{For the general multi-path channel model, new near-field beam training schemes
		need to be developed to cater to the complicated beam pattern, which is more challenging and left for
		our future work.}
 \begin{equation}\label{near_steering}
		\mathbf{b}^{H}\left(\theta_{\rm u}, r_{\rm u}\right)=\frac{1}{\sqrt{N}}\left[e^{-\frac{\jmath 2 \pi(r^{(0)}-r_{\rm u})}{\lambda}}, \cdots, e^{-\frac{\jmath 2 \pi(r^{(N-1)}-r_{\rm u})}{\lambda}}\right],
\end{equation}
where $\theta_{\rm u}\triangleq2d\cos(\phi_{\rm u})/ \lambda\in[-1,1]$ is the \emph{spatial} angle at the BS, $\phi_{\rm u}$ is the \emph{physical} angle-of-departure (AoD) from the BS center to the user, $r_{\rm u}^{(n)}=\sqrt{r_{\rm u}^2+\delta_{n}^2d^2-2r_{\rm u}\theta_{\rm u} \delta_{n}d}$ is the range between the $n$-th BS antenna and the user. 

\underline{\bf Signal model:} 
Let $x\in \mathbb{C}$ denote the transmitted symbol by the BS with power $P$ and $\mathbf{v}\in \mathbb{C}^{N \times 1}$ represent the BS transmit beamforming vector based on power-efficient analog phase shifters \cite{8030501}. Then the received signal at the user is given by
\begin{align}\label{Eq:general_Sig}
	y_{\rm near}
	=\sqrt{N}h_{\rm u} \mathbf{b}^{H}(\theta_{\rm u},r_{\rm u})\mathbf{v}x+z,
\end{align}
where $z$ with power $\sigma^2$ is the effective noise consisting of the received additive white Gaussian noise (AWGN) and received signals from  NLoS paths. 
As such, if given perfect CSI information $(\theta_{\rm u},r_{\rm u})$,  the optimal XL-array BS beamforming can be easily obtained as $\mathbf{v}^* = \mathbf{b}(\theta_{\rm u},r_{\rm u}),$ and the maximum achievable rate at the user in bits/second/hertz (bps/Hz) is $R=\log_{2}\l(1+\frac{P N|h_{\rm u} |^2}{\sigma^2}\r)$.

\section{Benchmark Beam Training Schemes}\label{Sec:Bench}
In this section, we introduce two benchmark schemes for near-field beam training and discuss their main limitations.
\vspace*{-5pt}
\subsection{2D Exhaustive Search under Polar-Domain Codebook}
For near-field beam training, a dedicated \emph{polar-domain} codebook was proposed in \cite{Cui2022channel}, which includes a set of predefined beamforming codewords, each steered to a specific location in the polar-domain. Specifically, the spatial angle domain $[-1,1]$ is uniformly sampled as 
	$\theta_n=\frac{2 n-N+1}{N},  n\in\mathcal{N}.$
While for each sampled angle $\theta_n$, the range domain is \emph{non-uniformly} sampled according to 
\begin{equation}
	r_{s,n}=\frac{1}{s} \alpha_{\Delta}\left(1-\theta_{n}^2\right), \quad \forall s\in\mathcal{S}\triangleq\{1,2,3, \cdots S\},
\end{equation}
where   $\alpha_{\Delta}\triangleq\frac{N^2 d^2}{2 \lambda_c \beta_{\Delta}^2}$, $\beta_{\Delta}$ is a constant value ensuring sufficiently small column coherence of each two near-field steering vectors, and $S$ is the number of sampled ranges determined by the required column coherence. Then  the polar-domain codebook is constructed as 
\begin{equation}
\bar{\mathbf{W}}_{\rm{Pol}} = \left[\bar{\mathbf{W}}_1,\cdots,\bar{\mathbf{W}}_n,\cdots,\bar{\mathbf{W}}_{N} \right],
\end{equation} where $\bar{\mathbf{W}}_n\triangleq\left[\bar{\mathbf{w}}_{n,1}, \cdots,\bar{\mathbf{w}}_{n,s},\cdots \bar{\mathbf{w}}_{n,S}\right]$ and $ \bar{\mathbf{w}}_{n,s}=  \mathbf{b}\left(\theta_n, r_{s, n}\right)$. Given the polar-domain codebook $\bar{\mathbf{W}}_{\rm{Pol}}$, a straightforward beam training method is applying a \emph{2D exhaustive search} to find the best codeword among all polar-domain codewords that yields the maximum received signal power at the user. However, this method  requires a total number of $T^{\rm{(ex)}}=NS$ training symbols, which is prohibitively high when $N$ is large, hence leaving insufficient time for data transmission.

\vspace*{-5pt}
\subsection{Two-Phase Near-field Beam Training}
To reduce the training overhead of  2D exhaustive search, a  two-phase near-field beam training method was proposed in \cite{Zhang2022fast}, which \emph{sequentially} estimates the user angle and range in two phases.
Specifically, in the first phase, the conventional DFT codebook is used for beam sweeping, which is defined as $\tilde{\mathbf{W}}_{\rm{DFT}}=\{\tilde{\mathbf{w}}_1, \cdots,\tilde{\mathbf{w}}_n,\cdots, \tilde{\mathbf{w}}_{N}\}$ where 
\begin{align}\label{Eq:farDFT}
\tilde{\mathbf{w}}_n&=\mathbf{a}(\theta_n)\nn\\
&\triangleq
\frac{1}{\sqrt{N}}\left[1,e^{-\jmath \pi\theta_{n}},\cdots,e^{-\jmath \pi(N-1)\theta_{n}}\right], \forall n\in\mathcal{N}.
\end{align}

With received signal powers at the user, a shortlist of candidate user angles can be estimated according to a key result that the user (spatial) angle is approximated located in the middle of an angular support (which shall be explicitly defined in Section~\ref{Sec:beampattern}). Then in the second phase, the user range is estimated by using the polar-domain codebook to find the best user range given the candidate user angles. This two-phase beam training method entails $T^{(\rm 2P)}=N+KS$ training symbols, where $K$ denotes the number of candidate user angles. However, this method may not provide accurate user range information when $S$ is small, due to the on-grid range estimation. On the other hand, when  $S$ is large, this method still yields relatively high training overhead. 

To address the above issues, we propose two new and efficient near-field beam training schemes using  conventional DFT codebook, which will be shown to achieve more accurate user range estimation as well as lower beam training overhead than the two benchmark schemes.

\vspace*{-5pt}
\section{Near-Field Received Beam Pattern Analysis}\label{Sec:beampattern}
To obtain useful insights into the near-field beam training design under the DFT codebook, we first characterize the received beam pattern at the user when far-field beamforming vectors spanning in the entire angular domain are used for beam scanning. Interestingly, we show that the received beam pattern in the near-field contains both user angle and range information.
\vspace*{-5pt}
\subsection{Received Beam Pattern}
Let ${\tilde{\mathbf{w}}}(\Omega)=\mathbf{a}(\Omega)$ denote a far-field (oriented) beamforming vector, where $\Omega\in[-1,1]$ represents the spatial angle in the entire angular domain. Note that ${\tilde{\mathbf{w}}}(\theta_n)$ in \eqref{Eq:farDFT} can be regarded as a special case of ${\tilde{\mathbf{w}}}(\Omega)$, for which $\Omega=\theta_n$ with $n\in\mathcal{N}$ is the discrete sampled angle.  Consider the near-field user with a channel steering vector $\mathbf{b}^H(\theta_{\rm u}, r_{\rm u})$. When far-field beamforming vectors $\{{\tilde{\mathbf{w}}}(\Omega), \forall \Omega\in[-1,1]\}$ are applied for beam scanning, the received useful signals at the user (without noise taken into account) are obtained as
\begin{equation}
\eta(\Omega) = \sqrt{N}h_{\rm u} \mathbf{b}^{H}(\theta_{\rm u},r_{\rm u})\tilde{\mathbf{w}}(\Omega),~~~ \forall \Omega\in [-1,1].
\end{equation}
Therefore, the received signal power $|\eta(\Omega)|^2$ is fundamentally determined by $\mathbf{b}^{H}(\theta_{\rm u},r_{\rm u})\tilde{\mathbf{w}}(\Omega)$. To characterize its property, we first give the following definitions.

\begin{definition}\label{De:beam_gain}
\emph{Given a user channel steering vector $\mathbf{b}^H(\theta_{\rm u}, r_{\rm u})$, define $ f(\mathbf{w}, \theta_{\rm u}, r_{\rm u})\triangleq|\mathbf{b}^H(\theta_{\rm u}, r_{\rm u}) \mathbf{w}|$ as the normalized beam gain when a  beamforming vector $\mathbf{w}$ is applied. }

\end{definition}
\begin{definition}\label{De:domin}
\emph{Given a user channel steering vector $\mathbf{b}^H(\theta_{\rm u}, r_{\rm u})$, define $\mathcal{A}_{\beta{\rm dB}}(\theta_{\rm u}, r_{\rm u})$ as the  $\beta$dB \emph{angular support} when far-field beamformers $\{\tilde{\mathbf{w}}(\Omega), \forall \Omega\in[-1,1]\}$ are used for beam scanning. Specifically, $\mathcal{A}_{\beta{\rm dB}}(\theta_{\rm u}, r_{\rm u})$ characterizes a spatial angular region for which the  corresponding normalized beam gain is above $\kappa f_{\rm peak}$, where $\kappa=10^{-\frac{\beta}{20}}\in[0,1]$, and $f_{\rm peak}$ denotes the  peak normalized beam gain. Mathematically, 
\begin{align}\label{Eq:AdB}
	 &\mathcal{A}_{\beta{\rm dB}}(\theta_{\rm u}, r_{\rm u})\nn\\&= 
	 \left\{\Omega_0 \mid f\left( \tilde{\mathbf{w}}(\Omega_0), \theta_{\rm u}, r_{\rm u}\right)\!>\!\kappa \max _{\tilde{\mathbf{w}}(\Omega)} f\left(\tilde{\mathbf{w}}(\Omega), \theta_{\rm u}, r_{\rm u}\right)\right\}.
\end{align}
Moreover, its \emph{angular support width}  is defined as 
\begin{equation}
\Gamma_{\beta {\rm dB}}(\theta_{\rm u}, r_{\rm u})= \Omega_{\rm right}-\Omega_{\rm left},
\end{equation}
where $\Omega_{\rm left}$ and $\Omega_{\rm right}$ are respectively the smallest and largest spatial angle in $\mathcal{A}_{\beta{\rm dB}}(\theta_{\rm u}, r_{\rm u})$ within $[-1, 1]$.}
\end{definition}
\begin{remark} [Proper value of $ \beta $]\label{Remark:Proper beta}
	\emph{Note that it is necessary to set a proper $ \beta $ to guarantee the normalized beam gains within the angular support being no smaller than those at the boundaries (i.e., $ \Omega_{\rm left} $ and $ \Omega_{\rm right} $), since otherwise there may be multiple angular supports, thus affecting the proposed near-field range estimation scheme (See Section~V). To this end, we set $ \beta\approx $ 3 dB in this paper, which is widely used in the existing literature as a proper power-ratio criterion for determining the beam width/depth \cite{10304223,10273772} and shall be numerically verified efficient.}
\end{remark}

In Fig.~\ref{fig:width_distance}, we plot  the normalized beam gains of far-field beamformers  $\{\tilde{\mathbf{w}}(\Omega)\}$ versus the spatial angle $\Omega$ for different user angles and ranges $\{(\theta_{\rm u}, r_{\rm u})\}$.
 Two key observations are summarized  below.
 \begin{observation}\label{OB}\emph{In Fig.~\ref{fig:width_distance}, the received near-field beam pattern under far-field beamformers contains the following useful user angle and range information:
   \begin{itemize}
 \item[1)] \textbf{User angle versus middle of angular support:} The true user spatial angle approximately locates in the \emph{middle} of the $3$dB angular support, i.e., $$\theta_{\rm u}\approx {\rm{Med}}(\mathcal{A}_{3{\rm dB}}(\theta_{\rm u}, r_{\rm u})).$$
 \item[2)] \textbf{User range versus angular support width:} Given the user angle $\theta_{\rm u}$, the angular support width $\Gamma_{3{\rm dB}}(\theta_{\rm u}, r_{\rm u})$ \emph{decreases} with the user range $r_{\rm u}$. Besides, given the user range, the angular support $\Gamma_{3{\rm dB}}(\theta_{\rm u}, r_{\rm u})$ \emph{decreases} with the absolute spatial angle (i.e., $|\theta_{\rm u}|$). 
 \end{itemize}
}
 \end{observation}
 \begin{figure}
    \centering
   	\vspace{-10pt}
    \includegraphics[width=8.5cm]{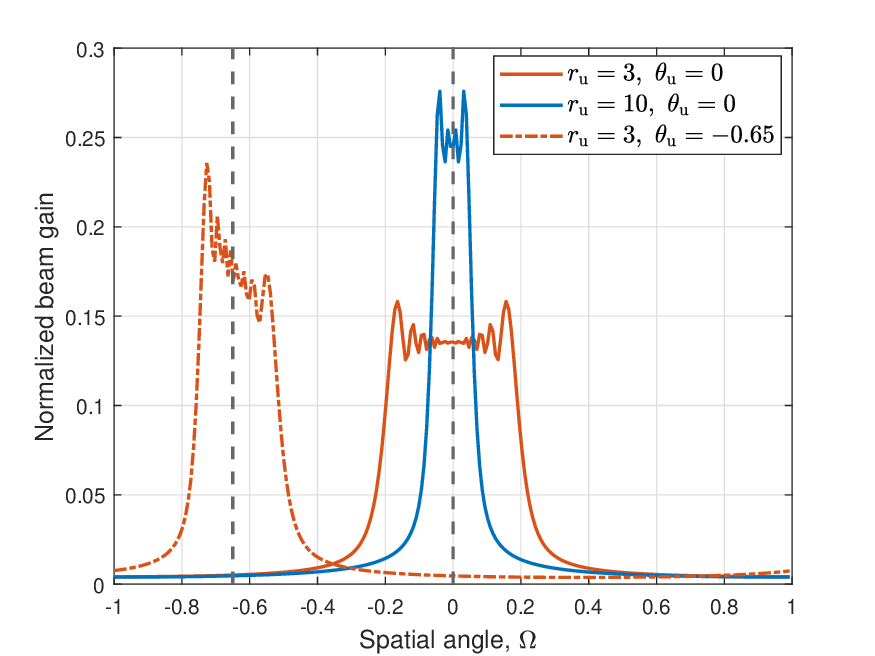}
    \caption{Received beam pattern under far-field beamformers for different user angles and ranges, where $N=256$ and $f=30$ GHz. The true user spatial angle is marked by the black dotted lines.}
    \label{fig:width_distance}
    \vspace{-10pt}
\end{figure}

Observation~\ref{OB} indicates that the user angle can be estimated by determining the middle of the angular support (which was originally revealed in \cite{Zhang2022fast}). Moreover, based on estimated user angle, the user range then may be inferred from the angular support width, since the angular support width monotonically decreases with the user range. However, it is \emph{intractable} to obtain an explicit expression for the angular support width $\Gamma_{\beta {\rm dB}}(\theta_{\rm u}, r_{\rm u})$, because the peak of  normalized beam gains (i.e., $\max _{\tilde{\mathbf{w}}(\Omega)} f\left(\tilde{\mathbf{w}}(\Omega), \theta_{\rm u}, r_{\rm u}\right)$) can be numerically obtained only, hence making it difficult to obtain a closed-form expression. 
\vspace{-5pt}
\subsection{Surrogate Angular Support}
To tackle the above difficulty, we consider an alternative metric, called \emph{surrogate angular support} defined below,  which characterizes the angular support when the reference peak beam gain is replaced by another one when a far-field beamformer is steered towards the user angle. 

\begin{definition}\label{Def:SAS}\emph{Let $g(\theta_{\rm u}, r_{\rm u})\triangleq f(\tilde{\mathbf{w}}(\theta_{\rm u}), \theta_{\rm u}, r_{\rm u})$ denotes the normalized beam gain when a far-field beamforming vector is steered at the user angle, i.e,  $\tilde{\mathbf{w}}(\theta_{\rm u})$. Then, when far-field beamformers $\{\tilde{\mathbf{w}}(\Omega), \forall \Omega\in[-1,1]\}$ are used for beam scanning, $\widehat{\mathcal{A}}_{\beta{\rm dB}}(\theta_{\rm u}, r_{\rm u})$ is defined as the $\beta$dB (properly chosen as discussed in Remark~\ref{Remark:Proper beta}) surrogate angular support, where the normalized beam gains are above $\kappa g(\theta_{\rm u}, r_{\rm u})$. Mathematically,
\begin{align}\label{Eq:SAS}
 \widehat{\mathcal{A}}_{\beta{\rm dB}}(\theta_{\rm u}, r_{\rm u})\!=\! 
	 \left\{\widehat{\Omega}_0 \mid f( \tilde{\mathbf{w}}(\widehat{\Omega}), \theta_{\rm u}, r_{\rm u})>\kappa g(\theta_{\rm u}, r_{\rm u})\right\}.	 
\end{align} 
Moreover, the \emph{surrogate angular support width}  is defined as 
\begin{equation}\label{Eq:SASW}
\widehat{\Gamma}_{\beta {\rm dB}}(\theta_{\rm u}, r_{\rm u})= \widehat{\Omega}_{\rm right}-\widehat{\Omega}_{\rm left},
\end{equation}
where $\widehat{\Omega}_{\rm left}$ and $\widehat{\Omega}_{\rm right}$ are respectively the smallest and largest spatial angle in the surrogate angular support $\widehat{\mathcal{A}}_{\beta{\rm dB}}$ within $[-1, 1]$.   }
\end{definition}

Compared with the angular support given in Definition~\ref{De:domin}, the newly-defined surrogate angular support and its width  can be \emph{analytically characterized}, thanks to a closed-form expression for the normalized beam gain towards the user angle, i.e., $g(\theta_{\rm u}, r_{\rm u})$. In the following, we present several useful properties of surrogate angular support to shed useful insights into the relationship between the user range and the surrogate angular support width.
To this end,  a \emph{beam power ratio} is firstly defined, which is the received-beam-power ratio of far-field beamformers that steered to an arbitrary angle $\phi$ and the user angle $\theta_{\rm u}$,
\begin{equation}\label{Eq:G}
	\!\!\!\!\! G(\theta_{\rm u}, r_{\rm u}, \phi)\!\triangleq\!\frac{f^2\left(\tilde{\mathbf{w}}(\phi), \theta_{\rm u}, r_{\rm u}\right)}{f^2\left(\tilde{\mathbf{w}}(\theta_{\rm u}), \theta_{\rm u}, r_{\rm u}\right)}\!=\!\frac{\left|\mathbf{b}^H(\theta_{\rm u}, r_{\rm u}) \tilde{\mathbf{w}}(\phi)\right|^2}{\left|\mathbf{b}^H(\theta_{\rm u}, r_{\rm u}) \tilde{\mathbf{w}}(\theta_{\rm u})\right|^2}.
\end{equation}
Useful properties of the beam power ratio are given below.
\begin{lemma}\label{Th: domiant_region_relation}\emph{
		The function $G(\theta_{\rm u}, r_{\rm u}, \phi)$ can be approximated as  
		\begin{align}\label{Eq:Gmua}
			\!\!\!\! &G(\theta_{\rm u}, r_{\rm u}, \phi) \stackrel{(c_1,c_2)}{\approx} L(\mu, a)\nn\\
			&\triangleq\! \frac{\left[C\!\l(\!a \mu\!+\!\frac{N}{2\mu}\r)\!-\!C\!\l(\!a \mu\!-\!\frac{N}{2\mu}\r)\right]^2\!+\!\left[S\!\l(\!a \mu\!+\!\frac{N}{2\mu}\r)\!-\!S\!\l(\!a \mu\!-\!\frac{N}{2\mu}\r)\right] ^2}{4\ \left[C^2\l(\frac{N}{2\mu}\r) + S^2\l(\frac{N}{2\mu}\r)\right]}
		\end{align}
		where 
$C(\cdot)$ and $S(\cdot)$ denote the Fresnel integrals defined as $C(x)=\int_{0}^{x}\cos(\frac{\pi}{2}t^2) {\rm d} t$ and $S(x)=\int_{0}^{x}\sin(\frac{\pi}{2}t^2) {\rm d} t$, respectively, and the parameters $\mu$ and $a$ are given by
\begin{equation}\label{Eq:def_mu_a}
	\mu\triangleq\sqrt{\frac{r_{\rm u}}{d(1-\theta_{\rm u}^2)}}, ~~~~~a\triangleq\theta_{\rm u}-\phi.
\end{equation}  }
\end{lemma}

\begin{proof}
	 Similar to \cite{zhang2023mixed}, we get the approximated expression by applying two approximations to both numerator and denominator of beam power ratio function $G(\theta_{\rm u}, r_{\rm u}, \phi) $. The first approximation, denoted as $(c_1)$, adopts Taylor expansion and Fresnel approximation. The second approximation, denoted as $(c_2)$, transforms the summation form into integral form. Last, by defining $\mu$ and $a$ in \eqref{Eq:def_mu_a}, we get the final result in \eqref{Eq:Gmua}.
\end{proof}
\begin{figure}
	\centering
	\vspace{-10pt}
	\subfigure[$\theta_{\rm u}=0$]{
		\label{Approximation 1}
		\includegraphics[width=0.48\linewidth]{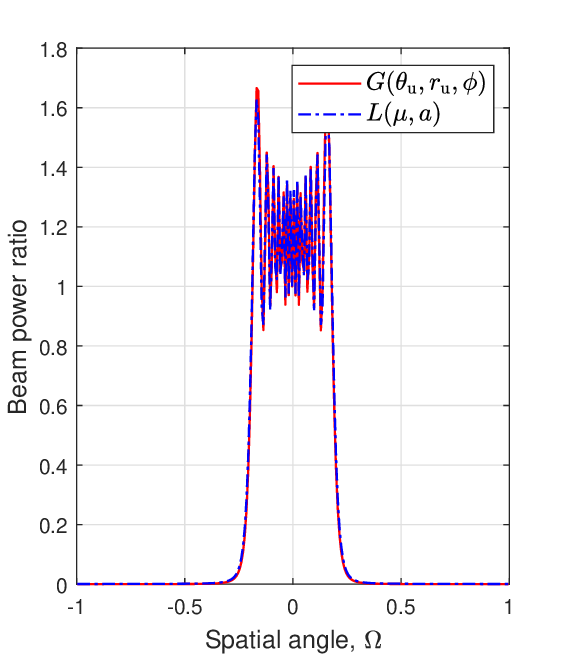}
	}\hspace{-10pt}
	\subfigure[$\theta_{\rm u}=-0.2$]{
		\label{Approximation 2}
		\includegraphics[width=0.48\linewidth]{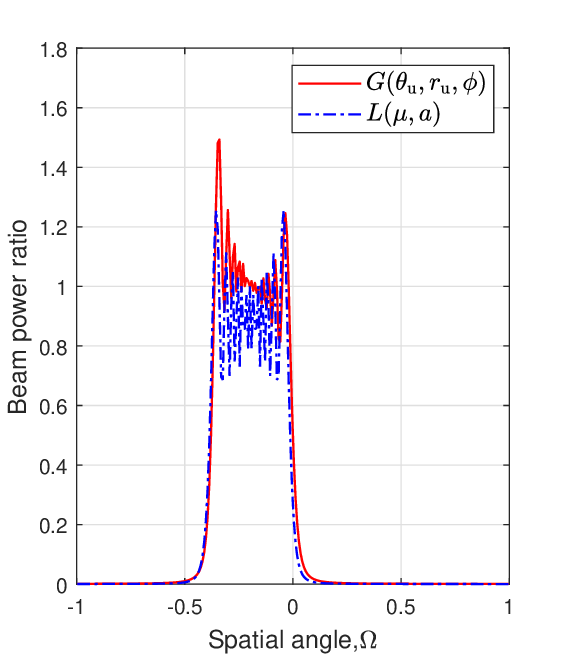}
	}
	
	\caption{The beam power ratio function $G(\theta_{\rm u}, r_{\rm u}, \phi)$ versus the approximation function $L(\mu, a)$. }
	\vspace{-10pt}
	\label{fig:fig_G_approximation}
\end{figure}
\begin{lemma}[Symmetricity of $L(\mu, a)$]\label{Le:Sym}\emph{Given $\theta_{\rm u}$ and $r_{\rm u}$, $L(\mu, a)$ is a function of $\phi$,  which is  symmetric function w.r.t. $\phi=\theta_{\rm u}$.}
\end{lemma}
\begin{proof}
First, we show $L(\mu, a)$ is a symmetric function w.r.t. $a=0$, i.e., $L(\mu, a)=L(\mu, -a)$. 
It can be easily shown that as $C(x)$ and $S(x)$ are symmetrical w.r.t. $x=0$, i.e. $C(-x) = -C(x), S(-x) = -S(x)$. Thus, we have $C(-a\mu-\frac{N}{2\mu})=- C(a\mu+\frac{N}{2\mu})$, $C(-a\mu+\frac{N}{2\mu}) = -C(a\mu-\frac{N}{2\mu})$. Similarly, we have $S(-a\mu-\frac{N}{2\mu})=- S(a\mu+\frac{N}{2\mu})$, $S(-a\mu+\frac{N}{2\mu}) = -S(a\mu-\frac{N}{2\mu})$. 
As such, 
\vspace{-3pt}
\begin{equation}
	\begin{aligned}
		&L(\mu,-a) \\ 
		&\!=\!\frac{\left[\!C\!\l(\!\!-a \mu\!\!+\!\!\frac{N}{2\mu}\!\r)\!\!-\!C\!\l(\!\!-a\mu\!\!-\!\!\frac{N}{2\mu}\!\r)\!\right]^2\!\!\!\!+\!\!\left[\!S\!\l(\!\!-a \mu\!\!+\!\!\frac{N}{2\mu}\!\r)\!\!-\!\!S\!\l(\!\!-a \mu\!\!-\!\!\frac{N}{2\mu}\!\r)\!\right]^2\!}{4\ \left[C^2\l(\frac{N}{2\mu}\r) + S^2\l(\frac{N}{2\mu}\r)\right]}\\
		& \!= \!\frac{\left[\!C\!\l(\!a \mu\!+\!\frac{N}{2\mu}\r)\!\!-\!C\!\l(\!a \mu\!-\!\frac{N}{2\mu}\!\r)\!\right]^2\!\!+\!\left[\!S\!\l(\!a \mu\!+\!\frac{N}{2\mu}\r)\!\!-\!\!S\!\l(\!a \mu\!-\!\frac{N}{2\mu}\!\r)\!\right]^2\!}{4\ \left[C^2\l(\frac{N}{2\mu}\r) + S^2\l(\frac{N}{2\mu}\r)\right]}\\
		& = L(\mu,a),
	\end{aligned}
\end{equation}
that is, $L(\mu,a)$ is a symmetric function w.r.t. $a = 0$.
Next, from (\ref{Eq:def_mu_a}), we have $a = \theta_{\rm u}-\phi$, which is a linear function of $\phi$ given $\theta_{\rm u}$. Thereby, $L(\mu,a)$ is a function of $\phi$ and is symmetric w.r.t $\theta_{\rm u}-\phi=0$. Then, the function $L(\mu,a)$ is symmetric w.r.t $\phi = \theta_{\rm u}$, thus completing the proof. 
\end{proof}

In Fig.~\ref{fig:fig_G_approximation}, we plot the curves for both functions of  $G(\theta_{\rm u}, r_{\rm u}, \phi)$ and  $L(\mu, a)$ versus $\phi$, with $N=256$, and $d=\frac{\lambda}{2}$. First, it is observed that when $\theta_{\rm u}=0$, $L(\mu, a)$ is almost the same as $G(\theta_{\rm u}, r_{\rm u}, \phi)$, which indicates that the approximation in \eqref{Eq:Gmua} is highly accurate. Next, when $\theta_{\rm u}\neq 0$, the approximated beam power ratio $L(\mu, a)$ does not match $G(\theta_{\rm u}, r_{\rm u}, \phi)$ very well in the surrogate angular support. In particular, $G(\theta_{\rm u}, r_{\rm u}, \phi)$ is \emph{asymmetric} w.r.t. $\phi=\theta_{\rm u}$, while $L(\mu, a)$ is \emph{symmetric} w.r.t. $\phi=\theta_{\rm u}$ (as analytically shown in Lemma~\ref{Le:Sym}). However, one can observe that the surrogate angular support width based on $G(\theta_{\rm u}, r_{\rm u}, \phi)$ and its approximation $L(\mu, a)$ are quite similar.  This thus motivates us to obtain an approximated surrogate angular support width based on the more tractable function $L(\mu, a)$, which is given below.

\begin{proposition}\label{Prop:SASW} \emph{Consider near-field beam training for a user with  the user angle and range $\{\theta_{\rm u},r_{\rm u}\}$. Given a properly chosen $ \beta $ as in Remark 1 (e.g., $ \beta $=3 dB) such that $L(\mu, a)$ is monotonic w.r.t. parameter $a$,  the surrogate angular support can be approximated as
			\begin{equation}
				\widehat{\mathcal{A}}_{\beta{\rm dB}}(\theta_{\rm u}, r_{\rm u}) \approx[\theta_{\rm u}-a_0, \theta_{\rm u}+a_0],
			\end{equation}
			where $a_0>0$ satisfies $L(\mu, a_0)=10^{-\frac{\beta}{20}}$, and $\mu=\sqrt{\frac{r_{\rm u}}{d(1-\theta_{\rm u}^2)}}$. 
			The corresponding surrogate angular support width is approximated as $$\widehat{\Gamma}_{3 {\rm dB}}(\theta_{\rm u}, r_{\rm u})\approx2a_0.$$}
\end{proposition}
\vspace{8pt}
\begin{proof}
		First, when $  \beta $ is properly chosen (e.g., $\beta=3$ dB)  such that $ L(\mu, a) $ is monotonic w.r.t. parameter $ a $,  it can be verified by enumerating all possible $a>0$ that there exists one and only one $a_0$ satisfying $L(\mu, a_0) = 10^{-\frac{\beta}{20}}$.
		Moreover, $L(\mu,a)$ is symmetrical w.r.t. $a=0$. Thus, for $a_0>0$ satisfying $L(\mu, a_0) = 10^{-\frac{\beta}{20}}$, we always have $L(\mu,- a_0) = 10^{-\frac{\beta}{20}}$.
		Next, $L(\mu,a)$ is an approximation of the beam power ratio function  $G(\theta_{\rm u}, r_{\rm u}, \phi)$, so that the obtained surrogate angular support width $\widehat{\Gamma}_{3 {\rm dB}}(\theta_{\rm u}, r_{\rm u})\approx2a_0$, thus completing the proof.
\end{proof}
It is worth noting that it is intractable to determine an optimal $\beta$ that can guarantee the monotonicity of function $ L(\mu, a) $ w.r.t. parameter $a$ as well as achieve the best performance. 
	Thus, we resort to a commonly used dominant power-ratio criterion, i.e., $ \beta=3 $ dB, which shall be numerically verified to attain high joint estimation accuracy and achievable rate. 
	Moreover, it is worth noting that 3 dB power-ratio criterion have been widely used in the literature for determining the angular/range interval with dominant power, such as the 3 dB beam width \cite{10304223} and 3 dB beam depth \cite{10273772}; which is thus also selected in this paper.
Proposition~\ref{Prop:SASW} indicates  that the surrogate angular support width is jointly determined by the user angle, user range, and the antenna spacing. 			

\begin{corollary}\label{Cor:rRayl}\emph{When $r_{\rm u}\to Z_{\rm Rayl}$, we have $\widehat{\Gamma}_{\beta {\rm dB}}(\theta_{\rm u}, r_{\rm u})  < \frac{2\lambda}{N}$.}
\end{corollary}
\begin{proof}
As $r_{\rm u}$ approaches $Z_{\rm Rayl}$, the near-field channel reduces into the far-field channel. As such, we have  channel steering vector $\mathbf{b}^{H}(\theta_{\rm u},r_{\rm u}) \to$  $\mathbf{a}^H(\theta_{\rm u})$. Upon substituting $\mathbf{a}^H(\theta_{\rm u})$ into \eqref{Eq:G}, the beam power ratio function $G$ becomes
\begin{align}
	G(\theta_{\rm u},r_{\rm u},\phi) &\approx \frac{\left | \mathbf{a}^H(\theta_{\rm u})\tilde{\mathbf{w} }(\phi )   \right |^2 }{\left | \mathbf{a}^H(\theta_{\rm u})\tilde{\mathbf{w} }(\theta_{\rm u} )  \right |^2 }= \left | \frac{1}{N} \mathbf{a}^H(\theta_{\rm u})\tilde{\mathbf{w} }(\phi ) \right |^2  \nn\\& =   \left | \frac{1}{N}  \sum_{n\in \mathcal{N} }^{} e^{-\jmath\frac{2\pi}{\lambda }n(\theta_{\rm u}-\phi) }  \right |^2 \\
	&\stackrel{(c_3)}{=}\left | \frac{\sin(\frac{\pi}{\lambda }N\Delta_{\theta_{\rm u}} ) }{N\cdot \sin(\frac{\pi }{\lambda } \Delta_{\theta_{\rm u}})}  \right | ^2,
\end{align}
where $(c_3)$ is obtained by letting $\Delta_{\theta_{\rm u}} = \theta_{\rm u}-\phi$ and applying the sum formula of geometric series.

The term $\frac{\sin(\frac{\pi}{\lambda }N\Delta_{\theta_{\rm u}} ) }{N\cdot \sin(\frac{\pi }{\lambda } \Delta_{\theta_{\rm u}})} $ corresponds to the Dirichlet kernel function \cite{li2023nearfield}. It is equal to $0$ when $\frac{\pi}{\lambda }N\Delta_{\theta_{\rm u}} = \pm \pi$, e.g., $\Delta_{\theta_{\rm u}} = \pm \frac{\lambda}{N}$.
Therefore, the null to null beam width of the main lobe is $ \frac{2\lambda}{N}$. 
Consequently, the width of the $\beta$dB surrogate angular support is smaller than $\frac{2\lambda}{N}$, thus completing the proof.
\end{proof}

Corollary~\ref{Cor:rRayl} shows that when the user range is larger than the Rayleigh distance, there is a very small (non-zero) angular support width since it reduces to the far-field scenario and the energy-spread effect disappears in this region.

In Fig.~\ref{fig:width_versus_distance_angle}, we numerically show the surrogate angular support width versus the user range for different user spatial angles. It is observed that $\widehat{\Gamma}_{3{\rm dB}}(\theta_{\rm u}, r_{\rm u})$ is generally decreasing with the user range and eventually converges to a very small value when the user is sufficiently far away from the XL-array. More specifically, the decreasing rate of $\widehat{\Gamma}_{3{\rm dB}}(\theta_{\rm u}, r_{\rm u})$ is fluctuating in the range domain. Generally, $\widehat{\Gamma}_{3{\rm dB}}(\theta_{\rm u}, r_{\rm u})$ drops quickly when the user range is small and slowly vice versa. 
However, in some \emph{range-non-sensitive} region (e.g., the user range from 18 m to 23 m for $\theta_{\rm u}=0$, $N=256$, and $f=30$ GHz as shown in Fig.~\ref{fig:width_versus_distance_angle} ), the surrogate angular support width decreases very slowly. It is worth noting that these\textit{ range-non-sensitive} regions are dependent on the antenna number $N$ and carrier frequency $f$.
Last, given the same user range, $\widehat{\Gamma}_{3{\rm dB}}(\theta_{\rm u}, r_{\rm u})$ decreases with the (absolute) user spatial angle, which is consistent with the results in Fig.~\ref{fig:width_distance}.
\begin{figure}
	\centering
	\vspace{-10pt}
	\includegraphics[width=8.5cm]{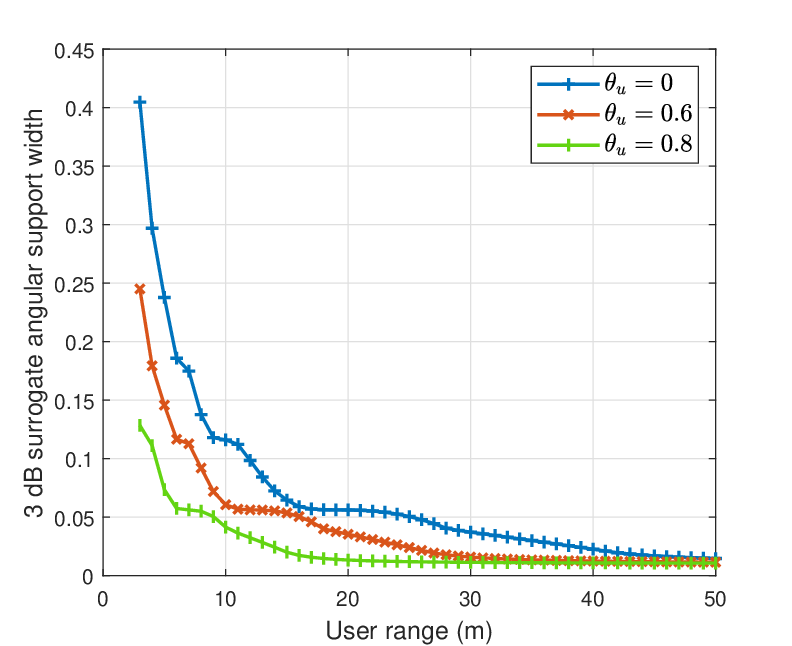}
	\caption{3dB surrogate angular support width versus user range for three user angles, where $N=256$ and $f=30$ GHz.}
	\vspace{-10pt}
	\label{fig:width_versus_distance_angle}
\end{figure}

\vspace*{-5pt}
\section{Proposed joint angle and range estimation with DFT codebook}
In this section, we propose two efficient schemes to jointly estimate the user angle and range with conventional DFT codebook. Specifically, the first scheme leverages the surrogate angular support width to estimate the user range, while the second scheme directly utilizes the received power ratios in the surrogate angular support for user range estimation.
\vspace*{-5pt}
\subsection{Scheme 1: Angle Support Width based Joint Angle and Range Estimation}
In Section~\ref{Sec:beampattern}, we show that the surrogate angular support width is a function of the user angle and range. This thus motivates us to first estimate the user angle and then leverage the surrogate angular support width to resolve the user range, which is termed as \emph{angle support width based joint angle and range estimation} (ASW-JE) scheme. To this end, we first obtain a semi-closed form expression for the user range according to the surrogate angular support width.

\begin{corollary}\label{Pro:scheme1}\emph{Based on Lemma~\ref{Th: domiant_region_relation}, when given $\theta_{\rm u}$ and $\phi$, the user range can be resolved from the $\beta$dB surrogate angular support width, which is given by
\begin{equation}\label{Eq:scheme1_estimate_range}
r_{\rm u}\approx \mu^2_0~ d(1-\theta_{\rm u}^2),
\end{equation}
where $\mu_0$ satisfies $L(\mu_0, a)=10^{-\frac{\beta}{20}}$, and $a=\theta_{\rm u}-\phi$.}
\end{corollary}
\begin{proof}
{First, it is observed from \eqref{Eq:Gmua} that $G(\theta, r, \phi)$ reduces to $L(\mu, a)$ by replacing $(\theta,r,\psi)$ with $(\mu,a)$. Second, given $\theta$ and $\psi$ (i.e., fixed $a$), $L(\mu, a)$ is now depends on $\mu$ only. Next, by enumerating all possible $\mu$, it is verified that we can obtain a $\mu_0$ that satisfies $L(\mu_0, a)=10^{-\frac{\beta}{20}}$. Then, the user range can be resolved from the relation between $\mu_0,d,\theta_{\rm u}$, as in \eqref{Eq:scheme1_estimate_range}, thus completing the proof.}
\end{proof}

\begin{remark}[Received beam pattern under DFT codebook]\label{Rem:DFTBeam}\emph{Note that unlike the received beam pattern for arbitrary far-field beamformers (i.e., ${\tilde{\mathbf{w}}}(\theta_n)$), the received beam pattern under the DFT codebook $\tilde{\mathbf{W}}_{\rm{DFT}}$ only consists of a finite number of \emph{sampled} normalized beam gains at selected sampled angles. Therefore, there may not exist an exactly $3$dB surrogate angular support for the case of DFT codebook. However, with the calculated normalized beam gains, we can choose a proper $\beta\approx 3$ to obtain the $\beta$dB surrogate angular support and estimate the user range according to Corollary~\ref{Pro:scheme1}.}
\end{remark}

Based on Corollary~\ref{Pro:scheme1} and Remark~\ref{Rem:DFTBeam}, we propose an efficient scheme to jointly estimate the user angle and range by leveraging the surrogate angular support, which  consists of the following three procedures. 
\subsubsection{{\underline{\textbf{Beam sweeping}}}} Similar to the far-field beam training procedure \cite{Cui2022channel}, the XL-array BS sequentially
sends $N$ training symbols, while it dynamically tunes beam directions (specified by  beam codewords) according to the predefined DFT codebook $\tilde{\mathbf{W}}_{\rm{DFT}}$. For each time, the received signal power at the user is given by 
\begin{equation}
p_n(\tilde{\mathbf{w}}_n) = | \sqrt{N}h_{\rm u} \mathbf{b}^{H}(\theta_{\rm u},r_{\rm u})\tilde{\mathbf{w}}_nx+z_n|^2, ~~~\forall n\in\mathcal{N}.
    \end{equation}
\subsubsection{\underline{\textbf{Joint angle and range estimation}}}
With the received signal powers $\{p_n(\tilde{\mathbf{w}}_n), \forall n\in\mathcal{N}\}$, the user angle and range can be sequentially estimated as follows based on the  angular support and its surrogate version.
\begin{itemize}
\item {\textbf{Angle estimation:}}
With $\{p_n(\tilde{\mathbf{w}}_n),\forall n\in\mathcal{N}\}$, the user obtains the index set of codewords for which it receives sufficiently high signal power (i.e., within the angular support), which is given by 
	\begin{equation}\label{Eq:B}
	\mathcal{B} = \{n |  p_n(\tilde{\mathbf{w}}_n)  > \kappa^2  \max_{\tilde{\mathbf{w}}_n} p_n(\tilde{\mathbf{w}}_n) \},
	\end{equation}
	where $\kappa^2\approx 0.5$ and its specific value is determined by the received beam pattern (see Remark~\ref{Rem:DFTBeam}). 
Note that for the angular support, $\mathcal{B}$ is obtained from the received signal powers after beam sweeping, while for ease of analysis, $\mathcal{A}_{\beta{\rm dB}}(\theta_{\rm u}, r_{\rm u})$ in \eqref{Eq:AdB} is obtained from the normalized beam gain without taking into account the effects of path-loss and received noise.  Then, the index of  estimated user angle in the DFT codebook is obtained as $\bar{n}= \lceil\rm{Med}(\mathcal{B})\rceil$ and the user angle is estimated as 
\begin{equation}
\bar{\theta}_{\rm u}^{(1)}=\theta_{\bar{n}}.\label{Eq:AngEst}
\end{equation}
\item {\textbf{Range estimation:}} Given the estimated user angle $\bar{\theta}_{\rm u}^{(1)}$, the surrogate angular support is firstly constructed. To this end, we  introduce a \emph{received power ratio} below
\begin{align}\label{Power ratio G}
	\eta_{\bar{n},n} &= \frac{p_n(\tilde{\mathbf{w}}_n)}{p_n(\tilde{\mathbf{w}}_{\bar{n}})} = \frac{\left|\sqrt{N} h_{\rm u} \mathbf{b}^{H}(\theta_{\rm u}, r_{\rm u}) \tilde{\mathbf{w}}_n x+z_n\right|^2}{\left|\sqrt{N} h_{\rm u} \mathbf{b}^H(\theta_{\rm u}, r_{\rm u}) \tilde{\mathbf{w}}_{\bar{n}} x+z_{\bar{n}}\right|^2} \nn\\&= \frac{\left|\mathbf{b}^{H}(\theta_{\rm u}, r_{\rm u}) \tilde{\mathbf{w}}_n+z^{\prime }_n\right|^2}{\left| \mathbf{b}^H(\theta_{\rm u}, r_{\rm u}) \tilde{\mathbf{w}}_{\bar{n}} +z^{\prime }_{\bar{n}}\right|^2}\nn\\&\overset{(c_4)}{\approx}\frac{\left|\mathbf{b}^{H}(\theta_{\rm u}, r_{\rm u}) \tilde{\mathbf{w}}_n\right|^2}{\left| \mathbf{b}^H(\theta_{\rm u}, r_{\rm u}) \tilde{\mathbf{w}}_{\bar{n}} \right|^2}\overset{(c_5)}{\approx} G(\theta_{\rm u}, r_{\rm u}, \theta_n),\forall n,
\end{align}
where $z_n^{\prime } = \frac{z_n}{\sqrt{N}h_{\rm u}}$, $z^{\prime }_{\bar{n}}= \frac{z_{\bar{n}}}{\sqrt{N}h_{\rm u}}$, the approximation in $(c_4)$ holds when $z_n^{\prime }=0$ and $z^{\prime }_{\bar{n}}=0$,  and the approximation in $(c_5)$ holds when $\theta_{\bar{n}}=\theta_{\rm u}$. Hence, in the range estimation, $\eta_{\bar{n},m}$ can be treated as an indicator/approximation for the beam power ratio $G(\theta_{\rm u}, r_{\rm u}, \theta_n)$. Then, the user  determines  $\phi=\theta_{m}$ to construct the surrogate angular support such that  
\begin{equation}\label{Eq:recpowratio}
\eta_{\bar{n},m}=\frac{p_n(\tilde{\mathbf{w}}_m)}{p_n(\tilde{\mathbf{w}}_n)}\approx  \frac{1}{2}.
\end{equation}
As such, by using Corollary~\ref{Pro:scheme1}, the user range can be  estimated  according to 
\begin{equation}\label{Eq:RanEst}
\bar{r}_{\rm u}^{(1)}=\mu^2_0 d(1-(\bar{\theta}^{(1)}_{\rm u})^2),
\end{equation}
where $\mu_0$ satisfies $L(\mu_0, a)=\eta_{\bar{n},m}$ and $a=\bar{\theta}^{(1)}_{\rm u}-\theta_{m}$. Note that $L(\mu, a)$ is a highly complicated function of $\mu$, hence the optimal solution to $L(\mu, a)=\eta_{\bar{n},m}$ can be obtained by applying an exhaustive search over $\mu$ in the set $\mathcal{Z}_{\mu}=\left\{\mu \mid \mu=\mu_{\min}, \mu_{\min }+\Delta \mu,  \mu_{\max }\right\}$, where $\mu_{\min}$, $\mu_{\max }$, and $\Delta \mu$ denotes the lower bound, upper bound and searching step size, respectively.
\end{itemize}

\begin{figure}[t] 
	\centering  
	\vspace{-0.5cm} 
	\subfigure[case 1 ]{
		\label{G versus $mu$ 1}
		\includegraphics[width=0.45\linewidth]{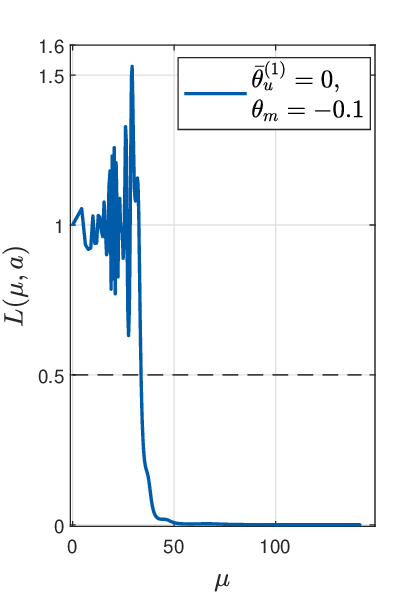}} \vspace{-10pt}
	\subfigure[case 2]{
		\label{G versus $mu$ 3}
		\includegraphics[width=0.45\linewidth]{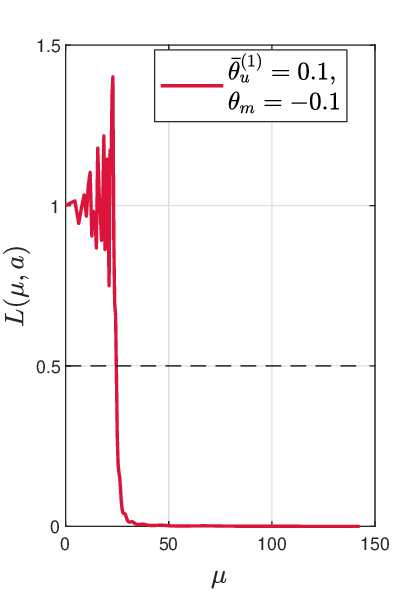}}
	\caption{Illustration of function $L(\mu, a)$ versus $\mu$, where the case 1 is $a = \bar{\theta}^{(1)}_{\rm u}-\theta_{m} = 0.1$, and the case 2 is $a = \bar{\theta}^{(1)}_{\rm u}-\theta_{m} = 0.2$. }
	\vspace{-5pt}
	\label{level}
\end{figure}

In Fig~\ref{level}, we plot the curves of function $L(\mu, a)$ versus $\mu$ for different $a$, where $N=256$, and $f=30$ GHz. It is observed that for different $a$, there is one specific $\mu_0$ ensuring $L(\mu_0, a)=0.5$. 
The above range estimation method shows a surprising  result that the DFT codebook, conventionally designed for angle estimation in the far-field, can also be used for estimating the user range in the near-field by effectively exploiting the near-field energy-spread effect.

\begin{remark}[Angle estimation error]\label{Re:angle_error}\emph{In practice, the accuracy of angle estimation in \eqref{Eq:AngEst} is affected by several factors. First, the estimated angle is obtained from  received signal powers in different time and hence is affected by the received noise (or equivalently received SNR). Thus, it is expected that the angle estimation is more accurate in the high-SNR regime (which shall be numerically verified in Section~\ref{SimuR}). Second, the angle estimation relies on an on-grid estimation method based on the DFT codebook. Therefore, its estimation accuracy is generally determined by the angle resolution of  DFT codebook.  A higher-resolution  DFT codebook in general leads to a more accurate angle estimation.}
\end{remark}
\begin{remark}[Range estimation error] \label{Re:range_error}\emph{It can be inferred from \eqref{Eq:RanEst} that  the range estimation error generally arises from two factors: 1) angle estimation error;  and 2) the error in obtaining $\mu_0$. On one hand, the range estimation requires the information of estimated user angle, and thus is significantly affected by the angle estimation error. Such error propagation will be more severe when the received SNR becomes lower. On the other hand, even with accurate estimated angle, the range estimation still hinges on the acquisition of $\mu_0$. Since $\mu_0$ is obtained by numerically solving the equation  $L(\mu_0, a)=\eta_{\bar{n},m}$, there is an additional approximation error in the low-SNR regime due to $L(\mu_0, a)\approx \eta_{\bar{n},m}$ in \eqref{Power ratio G}. 
Moreover, when the SNR is not sufficiently high, there may exist multiple solutions to \eqref{Eq:RanEst} for some cases. 
This is expected since there are some range-non sensitive regions as shown in Fig.~\ref{fig:width_versus_distance_angle}, where the surrogate angular support width does not reduce too much when the user range increases; thus, high-power noise may cause wrong estimation for $\mu$ in these regions.
}
\end{remark}

\subsubsection{\underline{\textbf{Beam determination}}}\label{SK1Sec:BeamDeter}
After the joint angle and range estimation with  DFT codebook, the user feedbacks $\{\bar{\theta}^{(1)}_{\rm u}, \bar{r}_{\rm u}^{(1)}\}$ to the XL-array BS. Then, we can apply a \emph{one dimensional (1D)-off-grid polar-domain codeword} approach to select the beam codeword for data transmission, as elaborated below.
First, we assume that the range of the polar-domain codeword can be flexibly controlled with an arbitrary value. 
Thus, the angle of the polar-domain codeword is on-grid, while the range of the polar-domain codeword is off-grid.
Therefore, given $\{\bar{\theta}^{(1)}_{\rm u}, \bar{r}_{\rm u}^{(1)}\}$, it can be easily obtained that the optimal XL-array beamforming vector for data transmission is $\mathbf{v}=\mathbf{b}(\bar{\theta}^{(1)}_{\rm u}, \bar{r}_{\rm u}^{(1)}).$

\begin{remark}[Improved scheme: Middle-$K$ angle selection]\label{Re:SelectionK}
	\emph{To further improve the accuracy of user angle and range estimation, one effective way is applying the \emph{middle-K angle selection scheme} proposed in \cite{Zhang2022fast}, which selects $K>1$ candidate user angles in the middle of  angular support rather than selecting one estimated angle only. Then, for each candidate angle $\bar{\theta}_{{\rm u},k}^{(1)}$, its corresponding candidate user range $\bar{r}_{{\rm u},k}^{(1)}$ can be estimated by following the similar procedure in \eqref{Eq:RanEst}. With all candidate user angles and ranges $\{\bar{\theta}_{{\rm u},k}^{(1)}, \bar{r}_{{\rm u},k}^{(1)}, \forall k\}$, another round of beam sweeping needs to be performed to select the best pair of candidate user angle and range. For example, for the data transmission scheme based on 1D-off-grid polar-domain codeword, $K$ training symbols should be sent to the user, while the XL-array dynamically changes its beamforming vectors according to $\mathbf{v}_k\!=\!\mathbf{b}(\bar{\theta}^{(1)}_{{\rm u}, k}, \bar{r}_{{\rm u}, k}^{(1)}), k=1,\!\cdots,\! K$. Finally, the best polar-domain beamforming for data transmission is obtained by comparing the received signal powers, i.e., $\mathbf{v}=\mathbf{b}(\bar{\theta}^{(1)*}_{\rm u}, \bar{r}_{\rm u}^{(1)*})=\arg\max_{\mathbf{v}_k}\{p_k(\mathbf{v}_k)\}.$ }
\end{remark}

\vspace*{-5pt}
\subsection{Scheme 2: Power-Ratio MSE based Joint Angle and Range Estimation}

The performance of the ASW-JE scheme is critically determined by the estimation accuracy of the (surrogate) angular support width, which may suffer considerable loss in the low-transmit-SNR regime.  To address this issue, we propose an alternative scheme,  called \emph{power-ratio MSE based joint angle and range estimation} (prMSE-JE), which greatly improves the range estimation accuracy.

The key idea is leveraging relatively high-power received signals, within the surrogate angular support, to estimate the user range via minimizing the sum power-ratio MSE. Thus, it can effectively average the received noise for achieving enhanced range estimation accuracy, when the transmit SNR is relatively low.

Similar to the ASW-JE scheme, the user angle can be estimated based on the angular support, i.e. $\mathcal{A} _{\beta {\rm dB}}(\theta_{\rm u},r_{\rm u})$, where $n$ and $\bar{n}$ denote index set of codeword within angular support and index of estimated user spatial angle $\bar{n}= \lceil\rm{Med}(\mathcal{B})\rceil$, respectively. Then, for the power-ratio, we can obtain the following based on \eqref{Power ratio G}
\begin{align}\eta_{\bar{n},n}& = \frac{\left|\mathbf{b}^{H}(\theta_{\rm u}, r_{\rm u}) \tilde{\mathbf{w}}_n+z^{\prime }_n\right|^2}{\left| \mathbf{b}^H(\theta_{\rm u}, r_{\rm u}) \tilde{\mathbf{w}}_{\bar{n}} +z^{\prime }_{\bar{n}}\right|^2}\nn\\&\approx\frac{\left|\mathbf{b}^{H}(\bar{\theta}_{\rm u}, r_{\rm u}) \tilde{\mathbf{w}}_n\right|^2}{\left| \mathbf{b}^H(\bar{\theta}_{\rm u}, r_{\rm u}) \tilde{\mathbf{w}}_{\bar{n}} \right|^2}\triangleq g_{\bar{n},n}(r_{\rm u}).\label{Eq:PRMSE}
\end{align}
Note that given $\{\bar{\theta}_{\rm u}, \tilde{\mathbf{w}}_n, \tilde{\mathbf{w}}_{\bar{n}}\}$, $g_{\bar{n},n}$ is  determined by the user range $r_{\rm u}$ only. On the other hand, $\eta_{\bar{n},n}$ can be obtained from the received signal powers at the user. Therefore, based on \eqref{Eq:PRMSE}, we propose to estimate the user range by minimizing the following sum power-ratio MSE
\begin{equation}
\sum_{n\in \mathcal{B}}|\eta_{\bar{n},n}-g_{\bar{n},n}(r_{\rm u})|^2.
\end{equation} The optimization problem can be formulated as follows
 \begin{subequations}
\begin{align}
\!\!\!\!({\bf P1}):\!\min_{r_{\rm u}}~  & \sum_{n\in \mathcal{B}}|\eta_{\bar{n},n}-g_{\bar{n},n}(r_{\rm u})|^2
\nn\\
\text{s.t.}~~
& Z_{\rm Fre}\le r_{\rm u}\le Z_{\rm{Rayl}}, \label{Eq:P1}
\end{align}
where $ Z_{\rm Fre}=0.5\sqrt{\frac{D^3}{\lambda}}$ and $Z_{\rm Rayl} $ denotes the Fresnel distance and Rayleigh distance\cite{you2023nearfield}, respectively. 
\end{subequations}

Note that $g_{\bar{n},n}(r_{\rm u})$ is a highly complicated function of $r_{\rm u}$ (which has a similar form with $ G(\theta_{\rm u}, r_{\rm u}, \phi)$ in \eqref{Eq:Gmua}), hence rendering $g_{\bar{n},n}(r_{\rm u})$ a non-convex function.  Although  (P1) is a non-convex optimization problem, it only involves one variable $r_{\rm u}$. Thus, the optimal solution to (P1) can be obtained by applying an exhaustive search over $r_{\rm u}$ in the set
$\mathcal{Z}_{r_{\rm u}}=\left\{r_{\rm u} \mid r_{\rm u}=r_{\min}, r_{\min }+\Delta r, \cdots, r_{\max }\right\},$
where $r_{\min }=Z_{\rm Fre}, r_{\max }= Z_{\rm Rayl}$, and $\Delta r$ represents the searching step size.  Based on the above, we introduce the main procedures for the prMSE-JE scheme.

\subsubsection{\underline{\textbf{Beam sweeping}}} The BS applies the conventional DFT codebook $\tilde{\mathbf{W}}_{\rm{DFT}}$  for beam sweeping, while the user receives $N$ signals with power  $p_n(\tilde{\mathbf{w}}_n), \forall n\in\mathcal{N}$.
%
\subsubsection{\underline{\textbf{Joint angle and range estimation}}} With the received signal powers $\{p_n(\tilde{\mathbf{w}}_n),\forall n\in\mathcal{N} \}$, the user angle and range are estimated as follows.
\begin{itemize}
	\item \textbf{Angle estimation:} Similar to the ASW-JE scheme, the user angle is estimated from the middle of angular support. Mathematically, the index set of the angular support is $\mathcal{B} = \{n |  p_n(\tilde{\mathbf{w}}_n)  > \kappa^2  \max_{\tilde{\mathbf{w}}_n} p_n(\tilde{\mathbf{w}}_n) \}$ with $\kappa^2\approx 0.5$  and the estimated user angle is $\bar{\theta}^{(2)}_{\rm u}=\theta_{\bar{n}}$.

	\item \textbf{Range estimation:} Given $\bar{\theta}^{(2)}_{\rm u}$, the surrogate angular support is constructed (similarly as in \eqref{Eq:recpowratio}). Then, the user range can be estimated by solving problem (P1) via the exhaustive search, which is denoted by $\bar{r}^{(2)}_{\rm u}$.  
\end{itemize}

\begin{remark}[Angle and range  estimation error]\emph{In general, the estimation error of the prMSE-JE scheme arises from two factors. The first one is the angle estimation error, which is affected by the received SNR, the resolution of DFT codebook, etc., as discussed in Remark \ref{Re:angle_error}. Moreover, the range estimation needs the information of estimated angle, thus its estimation  error is affected by the angle estimation error. In other words, the proposed prMSE-JE scheme can obtain high-accuracy range estimation if the angle estimation is accurate enough. However, in practice, the angle estimation error always exist due to the quantization of DFT codebook and environmental noises. To tackle this problem, \emph{middle-$K$ angle estimation} method, as elaborated in Remark~\ref{Re:SelectionK}, can be employed to improve the angle estimation accuracy.
Second, besides the received noise, the exhaustive search over $r_{\rm u}$ to solve the problem (P1) may also incur estimation error in the range sampling. Therefore, in general, there is a need to strike a \emph{complexity-accuracy} trade-off between minimizing the complexity of exhaustive search and maximizing the range estimation accuracy, while the solution performance is also affected by the received noise power.}
\end{remark}

\subsubsection{\underline{\textbf{Beam determination}}} Similar to the ASW-JE scheme, after the joint angle and range estimation with DFT codebook, the \emph{1D-off-grid polar-domain codeword} presented in Section~\ref{SK1Sec:BeamDeter} can be used for selecting the optimal beam codeword for data transmission. In this case, the designed XL-array beamforming vector is $\mathbf{v}=\mathbf{b}(\bar{\theta}^{(2)}_{\rm u}, \bar{r}_{\rm u}^{(2)})$, where $\bar{\theta}^{(2)}_{\rm u}$ is on-grid and $\bar{r}_{\rm u}^{(2)}$ is flexibly controlled in off-grid manner.
\vspace*{-5pt}
\subsection{Comparisons and Discussions}
Last, we compare the performance of proposed two near-field beam training schemes and the existing benchmark schemes in Section~\ref{Sec:Bench}, in terms of training overhead, estimation accuracy, and design complexity. Moreover, several extensions of this work are also presented.

\begin{itemize}
\item \textbf{Training overhead:} Both proposed schemes require a DFT-codebook based beam sweeping and a small number of training symbols for the middle-$K$ angle estimation and training. Thus, it requires $T^{(1)}=T^{(2)}=N+K$ training symbols only. This is much smaller than the two  benchmark schemes with $T^{\rm (2P)}=N+KS$ and $T^{\rm (ex)}=NS$. For example, consider the case where $N = 256$, $K=3$, and $S=5$. Then we have $T^{(1)}=T^{(2)} = 259$, which is smaller than $T^{\rm (ex)} = 1280$, $T^{\rm (2P)} = 271$.
\item \textbf{Angle/range estimation accuracy:} For angle estimation, the proposed two schemes achieve similar  estimation accuracy with the two-phase near-field beam training scheme \cite{Zhang2022fast}. While for range estimation, both proposed two schemes are able to achieve finer-grained range estimation. Compared with the ASW-JE scheme, the prMSE-JE scheme is expected to achieve \emph{more accurate} range estimation. This is because the prMSE-JE scheme utilizes all received signal powers in the angular support for minimizing the power-ratio MSE, rather than exploiting the surrogate angular support width solely as in the ASW-JE scheme. The detailed numerical comparison will be presented in simulation results in Section~\ref{SimuR}. 
\item \textbf{Design complexity:} For the proposed ASW-JE scheme, the overall computational complexity is determined by the exhaustive numerical search for $\mu_0$ in \eqref{Eq:RanEst}, which is  $\mathcal{O}(|\mathcal{Z}_{\mu}|N)$, where is $|\mathcal{Z}_{\mu}|$ is the cardinality of the searching collection over $\mu$. While for the prMSE-JE scheme, its computational complexity arises from the exhaustive search over $r_{\rm u}$, hence its complexity order is $\mathcal{O}(|\mathcal{Z}_{r_{\rm u}}|\cdot N \cdot |\mathcal{B}|)$, where $|\mathcal{Z}_{r_{\rm u}}|$ is the cardinality of the searching collection over $r_{\rm u}$, and $|\mathcal{B}|$ is the cardinality of angular support in \eqref{Eq:B}. Generally, the sampled ranges is set from the Fresnel distance to Rayleigh distance.
\end{itemize}

\begin{remark}[Multipath channels]\emph{For the multipath channel case, the extension of the proposed schemes can be divided into the following two cases.
		\begin{itemize}
			\item \textbf{LoS-dominant channel:} When the LoS channel component is dominant, the NLoS components can be treated as environment scatters. Thus, the proposed schemes still hold as the angle and range estimation hinge on the LoS channel component.
			\item \textbf{Comparable multi-path components:} For this case, the joint angle and range estimation is much more complicated, because there may arise overlapped angular supports caused by the LoS path and other NLoS components, which poses great challenges for angle estimation.
			Therefore, more complicated beam training schemes need to be developed for this case, which is left for our future work. 
		\end{itemize} 
}	
\end{remark}
\begin{remark}[Universal in both near- and far-field communications] \emph{The proposed joint angle and range estimation schemes  apply to both near- and far-field communications. Note that, the user angle can be estimated by finding the middle of defined angular support. 
As such, the angular support width can be employed to decide whether the user is located in the near-field or far-field. 
Specifically, if the angular support width is large, the user is located in the near-field region of XL-arrays, which thus can be exploited to estimate the user range by the proposed ASE-JE and prMSE-JE schemes. 
On the other hand, if the angular support width is small (e.g., only contains one DFT codeword), it indicates that the user is located in the far-field region. 
In this case, there is no need for subsequent range estimation.}  
\end{remark}

\vspace*{-5pt}
\section{Numerical results}\label{SimuR}
In this section, we provide numerical results to validate the effectiveness of our proposed ASW-JE and prMSE-JE near-field beam training schemes. The system parameters are set as follows. The XL-array  is equipped with $N=256$ antennas and operates at $f=30$ GHz frequency band. To ensure a sufficiently small column coherence of each two near-field steering vectors, we set  $\beta_{\Delta}=1.4$\cite{Cui2022channel}, which leads to $\alpha_{\Delta} = 41.80$. As such, the number of sampled ranges on each angle is $S = 5$. 
The reference path-loss is $\beta=(\lambda/4\pi)^2=-62$ dB and the (reference) SNR is defined as ${\rm SNR}=\frac{P\beta N}{r_{\rm u}^2\sigma^2}$ \cite{Zhang2022fast}, where the transmit power $P = 30 $ dBm and noise $\sigma^2=-70$ dBm, respectively.
The number of NLoS channel paths from the BS to the user is set to $L=2$.
 Moreover, we set the Rician factor $\kappa=30$ dB, which is practical for mmWave frequency bands \cite{10130629,10217152}.
Finally, all simulation results are carried out over $1000$ random channel realizations.

For performance comparison, we consider the following schemes.
\begin{itemize}
\item[1)] \textbf{Proposed ASW-JE scheme}: For this scheme, we consider both the cases with $K=1$ or $K=3$, which select $1$ and $3$ candidate estimated angles, respectively (see Remark~\ref{Re:SelectionK}). 
\item[2)] \textbf{Proposed prMSE-JE scheme:}  Similarly, for this scheme, we consider both the cases with $K=1$ or $K=3$. 
\item[3)] \textbf{Perfect-CSI based beamforming:} The beamforming vector is perfectly aligned with the channel steering vector, i.e., $ \mathbf{v}_{\rm perf}=\mathbf{b}(\theta_{\rm u},r_{\rm u})$.
 \item[4)] \textbf{Exhaustive-search beam training:} This scheme conducts an exhaustive-search over the polar-domain codebook, which consists of a 2D search over both the angle and range domains \cite{Cui2022channel}.
 \item[5)] \textbf{Two-phase beam training:} This scheme firstly estimates the user angle using the DFT codebook and then finds the user range using the polar-domain codebook \cite{Zhang2022fast}.
 \item[6)] \textbf{Far-field beam training:} This scheme directly selects the best codeword in  the DFT codebook that yields the maximum received signal power at the user \cite{9129778}.
\end{itemize}

We employ the following performance metrics to compare different schemes.
Specifically, the normalized mean square error (NMSE) is employed to evaluate the accuracy of estimated angle, which is defined as 
$\mathrm{NMSE}_{\rm angle}=\frac{\mathbb{E}\left(|\theta_{\rm u}-\bar{\theta}_{\rm u}|^2\right)}{\mathbb{E}\left(|\theta_{\rm u}|^2\right)},$ 
where $\theta_{\rm u}$ and $\bar{\theta}_{\rm u}$ denote the true user spatial angle and estimated user angle by different schemes, respectively. 
Similarly, the NMSE for range estimation is defined as 
$\mathrm{NMSE}_{\rm range}=\frac{\mathbb{E}\left(|r_{\rm u}-\bar{r}_{\rm u}|^2\right)}{\mathbb{E}\left(|r_{\rm u}|^2\right)},$
where $r_{\rm u}$ and $\bar{r}_{\rm u}$ denote the true and estimated user range by different schemes, respectively.
In addition, based on the signal model in (\ref{Eq:general_Sig}), the user achievable rate is given by 
$R=\log_{2}\l(1+\frac{|\mathbf{h}_{\rm near}^H \mathbf{v}|^2}{\sigma^2}\r),$
where the beamforming vector $\mathbf{v}$ is obtained by different beam training schemes. 
Moreover, in order to show the effect of beam training ovehead on the rate performance,  we also adopt another rate performance metric, namely, the effective achievable rate, which is  given by \cite{10130629}
$R_{\rm eff}=\left(1-\frac{T_{\rm tra}}{T_{\rm tot}} \right)\log_{2}\l(1+\frac{|\mathbf{h}_{\rm near}^H \mathbf{v}|^2}{\sigma^2}\r),$
where $T_{\rm tot}$ denotes the total number of symbols in each time frame and $T_{\rm tra}$ denotes the number of required training symbols (equivalently the training overhead). 
\vspace*{-5pt}
\subsection{Angle/Range Estimation NMSE}

First, Fig.~\ref{fig:NMSE_angle} shows the effect of reference SNR on the angle estimation accuracy with fixed user range $r_{\rm u}=12$ m and user angle $\theta_{\rm u}$ being randomly generated in the spatial domain $[-1,1]$. 
First, it is observed that the angle estimation NMSEs of proposed two schemes decrease with an increasing reference SNR under different $K$, and is close to that of the two-phase beam training schemes, while that of the exhaustive-search beam training scheme keeps unchanged. 
This is because under the DFT codebook, as SNR increases, the received beam pattern at the user generally become stabilized, thus causing smaller angle estimation error (see Remark~\ref{Re:angle_error}). 
In contrast, with the polar-domain codebook (which requires much more training symbols), the exhaustive-search beam training utilizes the beam-focusing property and selects the best polar-codeword that yields the maximum received signal power at the user, hence its angle estimation NMSE is robust to SNR. 
Second, in the high-SNR regime, by slightly increasing the number of candidate angles up to $K=3$, the proposed ASW-JE and prMSE-JE schemes attain much lower angle estimation NMSE and eventually achieve similar NMSE with that of the exhaustive-search beam training. 
This is because the selection of 3 candidate angles can further reduce the probability of angle estimation error as illustrated in Remark~\ref{Re:SelectionK}.

\begin{figure}[t!]
	\centering
	\vspace{-5pt}
	\includegraphics[width=7.5cm]{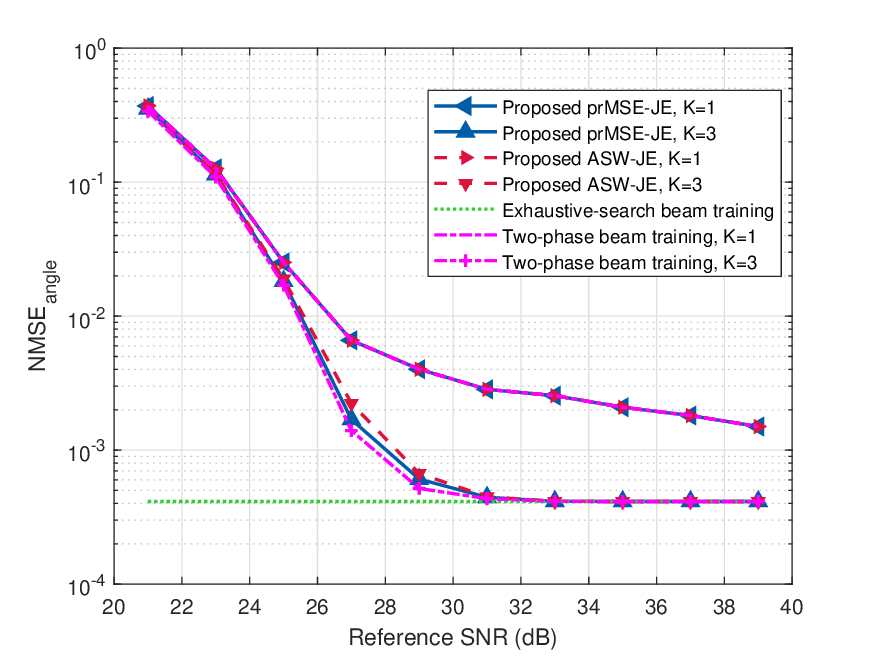}
	\caption{Angle estimation NMSE versus reference SNR.}
	\label{fig:NMSE_angle}
\end{figure}
\begin{figure}[t!]
	\centering
	\vspace{-5pt}
	\includegraphics[width=7.5cm]{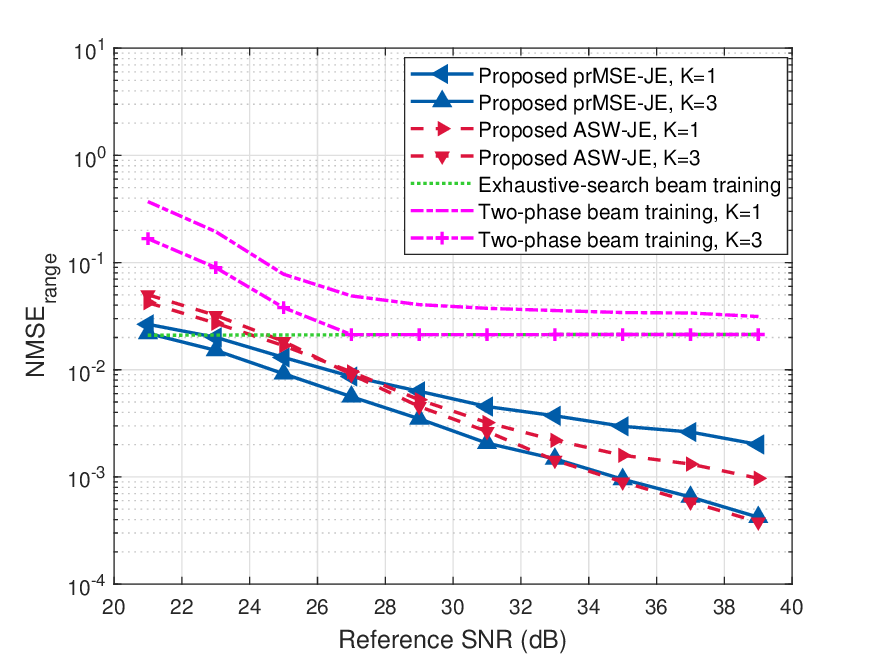}
	\caption{Range estimation NMSE versus reference SNR.}
	\label{fig:NMSE}
	\vspace{-10pt}
\end{figure}
Fig.~\ref{fig:NMSE} shows the curves of range estimation NMSE versus reference SNR. Several interesting observations are made as follows. First, the range estimation NMSE of the proposed two schemes decreases with the reference SNR, and they significantly outperform the two benchmark beam training schemes in the high-SNR regime. 
This is because the two benchmark beam training schemes employ predefined range codewords to estimate the user range; thus its accuracy is determined by the number of sampled ranges $S$. 
However, our proposed two schemes apply the off-grid range estimation (see \eqref{Eq:RanEst} and \eqref{Eq:P1}); hence achieving better resolution/accuracy in the high-SNR regime.
Second, it is observed that our proposed prMSE-JE and ASW-JE with $K=3$ achieve better range accuracy than those with $K=1$. 
This is because more candidate angles can attain more accurate angle as illustrated in Remark~\ref{Re:range_error}. Finally, it is observed that the performance of proposed prMSE-JE scheme is superior to the ASW-JE scheme. 
This is because the prMSE-JE scheme leverages all  received signal powers within the surrogate angular support (see \eqref{Eq:P1}), while the ASW-JE scheme only leverages the surrogate angular support width for range estimation (see \eqref{Eq:RanEst}), and thus may suffering inaccurate range estimation in lower-SNR regime.

\vspace*{-5pt}  
\subsection{Effect of System Parameters}
\begin{figure}[t!]
	\centering
	\vspace{-5pt}
	\includegraphics[width=7.5cm]{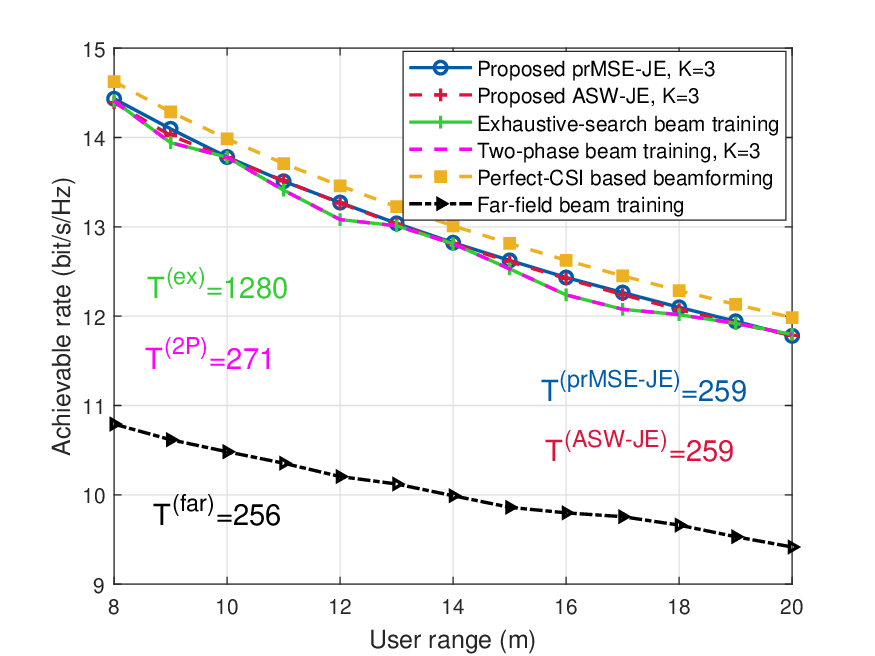}
	\caption{Achievable rate versus user range.}
	\label{fig:rate_versus_distance}
\end{figure}
\begin{figure}[t!]
	\centering
	\vspace{-5pt}
	\includegraphics[width=7.5cm]{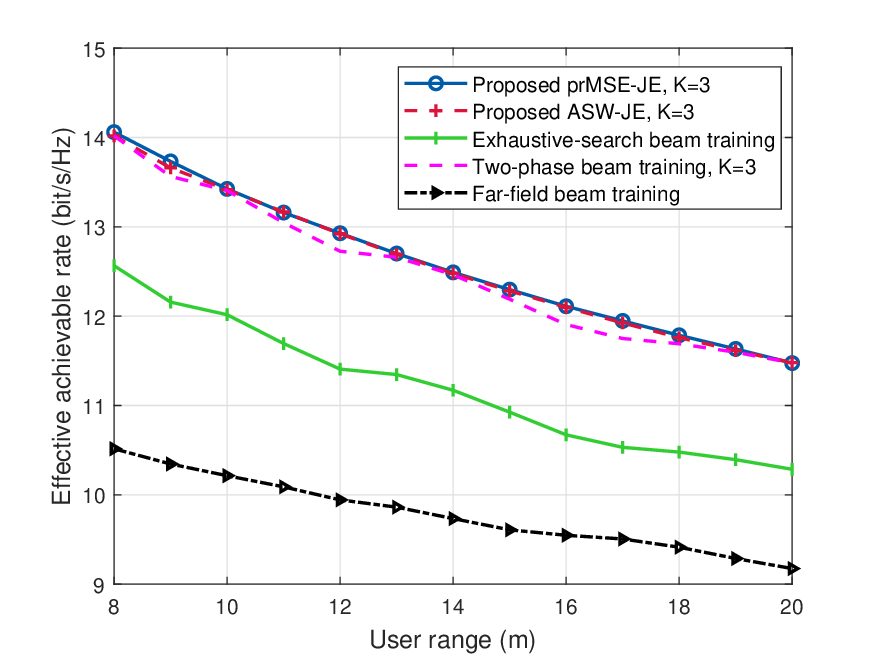}
	\caption{Effective achievable rate versus  user range.}
	\vspace{-10pt}
	\label{fig:eff_rate_versus_distance}
\end{figure}
\subsubsection{Effect of user range}
We plot in Fig.~\ref{fig:rate_versus_distance} the achievable rate performance versus the user range.
The user is randomly generated at the range from 8 m to 20 m, and at the spatial angle region $[-1,1]$.
First, it is observed that the proposed prMSE-JE scheme with $K=3$ exhibits better rate performance than the exhaustive-search beam training scheme. 
This improvement is attributed by the prominent energy-spread effect within this range region, hence leading to a larger angular support region that is beneficial to the ASW-JE and prMSE-JE schemes. 
Second, for all user ranges, the ASW-JE and prMSE-JE schemes with $K=3$ outperform the far-field beam training scheme, although they both use the DFT codebook for near-field beam training. 
This is because our proposed ASW-JE and prMSE-JE schemes effectively exploit the angle and range information embedded in the received beam pattern, while the conventional far-field beam training scheme simply selects the codeword yielding the maximum received signal power. 
Next, the proposed ASW-JE and prMSE-JE schemes achieve comparable achievable rate due to the fact that they have similar user angle estimation accuracy, while both obtain highly accurate user range estimation.
Finally, our proposed schemes attain more close rate performance with the perfect-CSI based beamforming scheme, while using much less beam training symbols than the exhaustive-search beam training  and the two-phase beam training schemes. 


In Fig.~\ref{fig:eff_rate_versus_distance}, we consider an alternative performance metric, effective achievable rate. 
First, it is observed that the proposed ASW-JE and prMSE-JE schemes with $K=3$ obtain better effective achievable rate than the two-phase beam training, and significantly outperform the exhaustive-search beam training scheme. 
This is because the proposed ASW-JE and prMSE-JE schemes greatly reduce the beam training overhead without sacrificing the achievable rate performance too much.
It is worth noting that the beam training overhead of the exhaustive-search beam training scheme is prohibitively high which is not suitable in practical wireless systems.

\begin{figure}[t!]
	\centering
	\vspace{-5pt}
	\includegraphics[width=7.5cm]{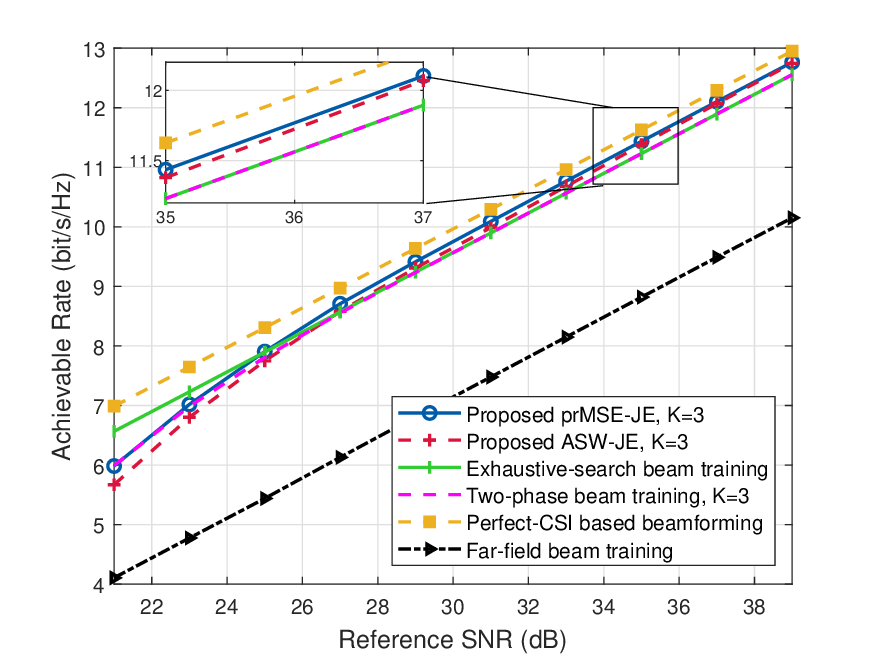}
	\caption{Achievable rate versus reference SNR.}
	\vspace{-5pt}
	\label{fig:rate_versus_snr}
\end{figure}
\begin{figure}[t!]
		\vspace{-5pt}
	\centering

	\includegraphics[width=7.5cm]{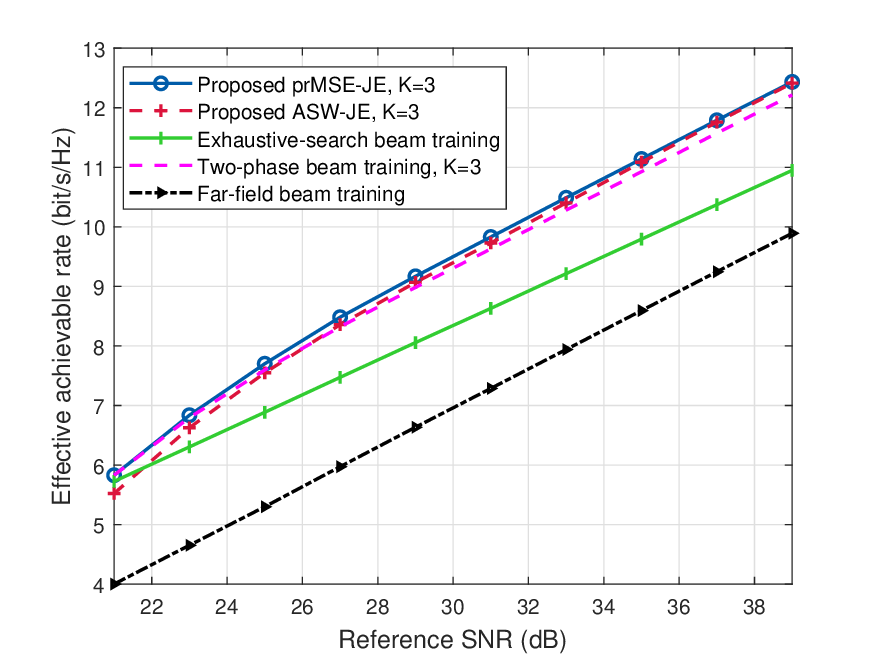}
	\caption{Effective achievable rate versus reference SNR.}

	\label{fig:eff_rate_versus_snr}
		\vspace{-10pt}
\end{figure}
\subsubsection{Effect of reference SNR}
In Fig.~\ref{fig:rate_versus_snr}, we compare the achievable rates of different schemes versus the reference SNR  given the user range $r_{\rm u}=16$ m with the user angle randomly distributed in region $[-1,1]$.
First, it is observed that both proposed ASW-JE and prMSE-JE schemes with $K=3$  outperform the exhaustive-search beam training scheme in the high-SNR regime (e.g., 28-38 dB).
This is because both proposed schemes estimate the user range from the received beam pattern in an off-grid manner. 
Second, in the low-SNR regime, the proposed ASW-JE and prMSE-JE schemes suffer a slightly achievable rate loss as compared to the exhaustive-search beam training scheme. 
This is because in the low-SNR regime, the received beam pattern fluctuates drastically, so that it may cause some error in the angle estimation (see Remark~\ref{Re:angle_error}). 
Next, the achievable rates of all schemes increase with the increasing reference SNR.
Finally, the proposed prMSE-JE scheme outperforms the proposed ASW-JE scheme due to the more robust and accurate user range estimation.

In Fig.~\ref{fig:eff_rate_versus_snr}, we plot the effective achievable rates of different schemes versus the reference SNR.
First, the proposed ASW-JE and  prMSE-JE schemes with $K=3$ obtain slightly rate improvement as compared to the two-phase beam training scheme and significant rate improvement over the exhaustive-search beam training scheme in the high-SNR regime. 
This is because the training overhead of proposed schemes is very small, while at the same time achieving better rate performance than the benchmark schemes. 
Second, it is observed that in the low-SNR regime, the rate performances of proposed schemes increase drastically with the increasing SNR. This is because the estimated angle and range become more accurate with a higher SNR.
Finally, one can observe that the rate performance of the prMSE-JE scheme is superior to the ASW-JE scheme due to more robust range estimation. 

\begin{figure}[t!]
	\centering
	\vspace{-5pt}
	\includegraphics[width=7.5cm]{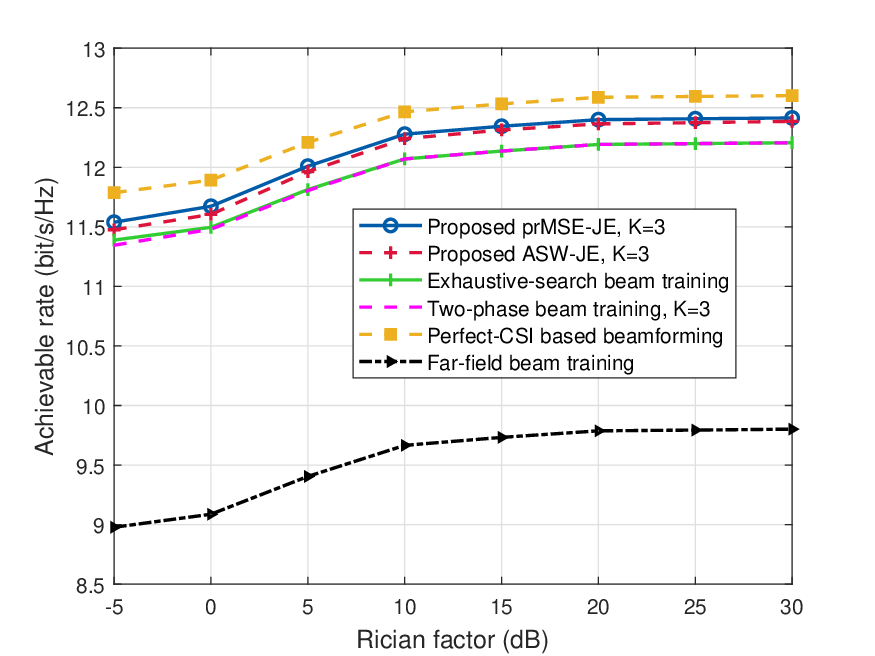}
	\caption{Achievable rate versus Rician factor.}
	\vspace{-5pt}
	\label{fig:rate_rician}
\end{figure}
\begin{figure}[t!]
	\vspace{-5pt}
	\centering
	\includegraphics[width=7.5cm]{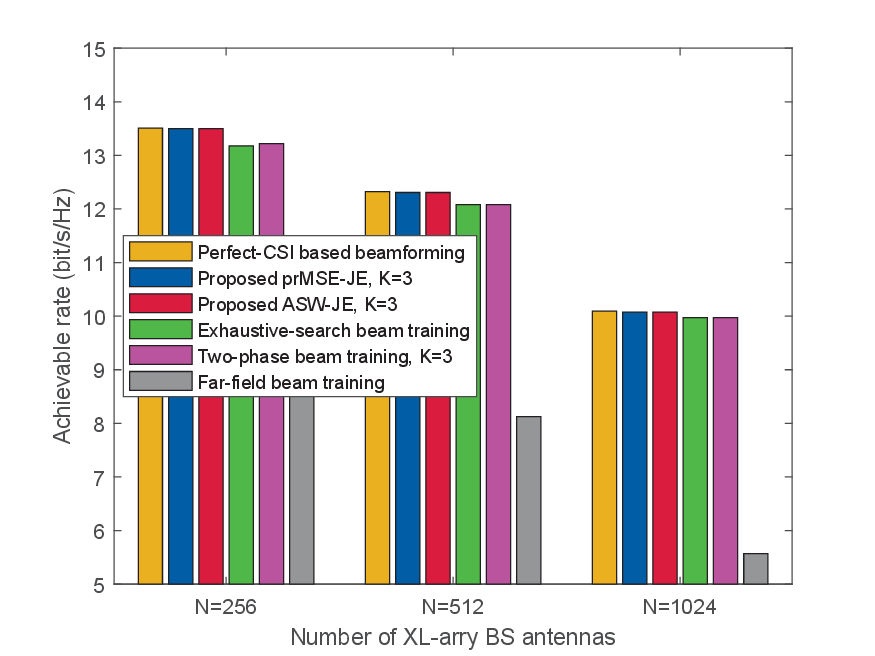}
	\caption{Achievable rate versus number of XL-array BS antennas.}
	\label{fig:rate_number}
	\vspace{-10pt}
\end{figure}

\subsubsection{Effect of Rician factors}
Next, we evaluate the effect of Rician factor on the achievable rate by different schemes in Fig.~\ref{fig:rate_rician}. 
First, it is observed that, as the Rician factor increases, the achievable rates of all schemes increase at first and eventually saturate when the Rician factor is sufficiently large.
Second, one can observe that the proposed ASW-JE and prMSE-JE schemes with $K=3$ outperform the exhaustive-search beam training scheme and two-phase beam training scheme due to the superior accuracy of user range estimation. 

\subsubsection{Effect of number of XL-array antennas}

Finally, we investigate the effect of number of BS antennas on the achievable rate in Fig.~\ref{fig:rate_number}.
Specifically, the user is assumed to be located at the Fresnel boundary, which increases rapidly with the number of antennas.
It is observed that the achievable rates of all schemes decrease with the number of antennas due to the larger path-loss.
Second, one can observe that our proposed schemes outperform the exhaustive-search beam training and two-phase beam training schemes under different number of antennas.
Last, the proposed schemes obtain considerable rate improvement compared to the far-field beam training scheme.

\vspace{-2pt}
\section{Conclusions}
In this paper, we studied efficient near-field beam training schemes using conventional DFT codebook.
To this end, we first show an interesting result that when far-field beamforming vectors are applied for beam scanning, the received beam pattern contains useful user angle and range information. 
Then, two new schemes, namely, ASW-JE and prMSE-JE, were proposed to jointly estimate the user angle and range by using the DFT codebook. Numerical results showed that both proposed schemes achieved lower training overhead and more accurate range estimation than existing benchmark schemes, while the prMSE-JE scheme is more robust in the low-SNR regime.

\bibliographystyle{IEEEtran}
\bibliography{IEEEabrv}

\begin{IEEEbiography}
	[{\includegraphics[width=1in, height=1.25in,clip, keepaspectratio]{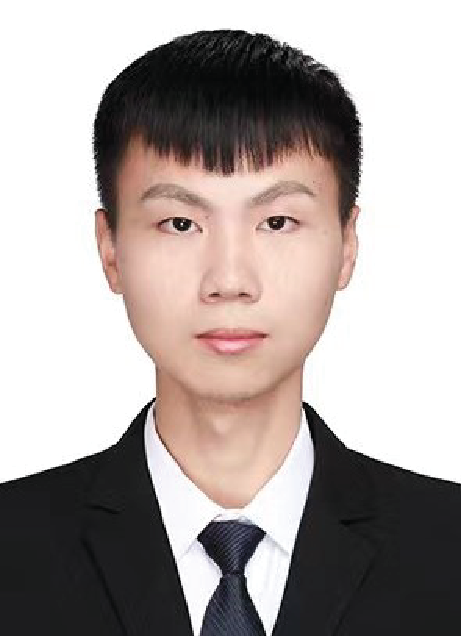}}]
	{Xun Wu} received the B.S. degree in communication engineering from Southern University of Science and Technology, Shenzhen, China, in 2022. He is currently pursuing the M.S. degree in Electronic Science and Technology from Southern University of Science and Technology, Shenzhen, China. His research interests include near-field communications and XL-MIMO communications.
\end{IEEEbiography}
\begin{IEEEbiography}
	[{\includegraphics[width=1in, height=1.25in,clip, keepaspectratio]{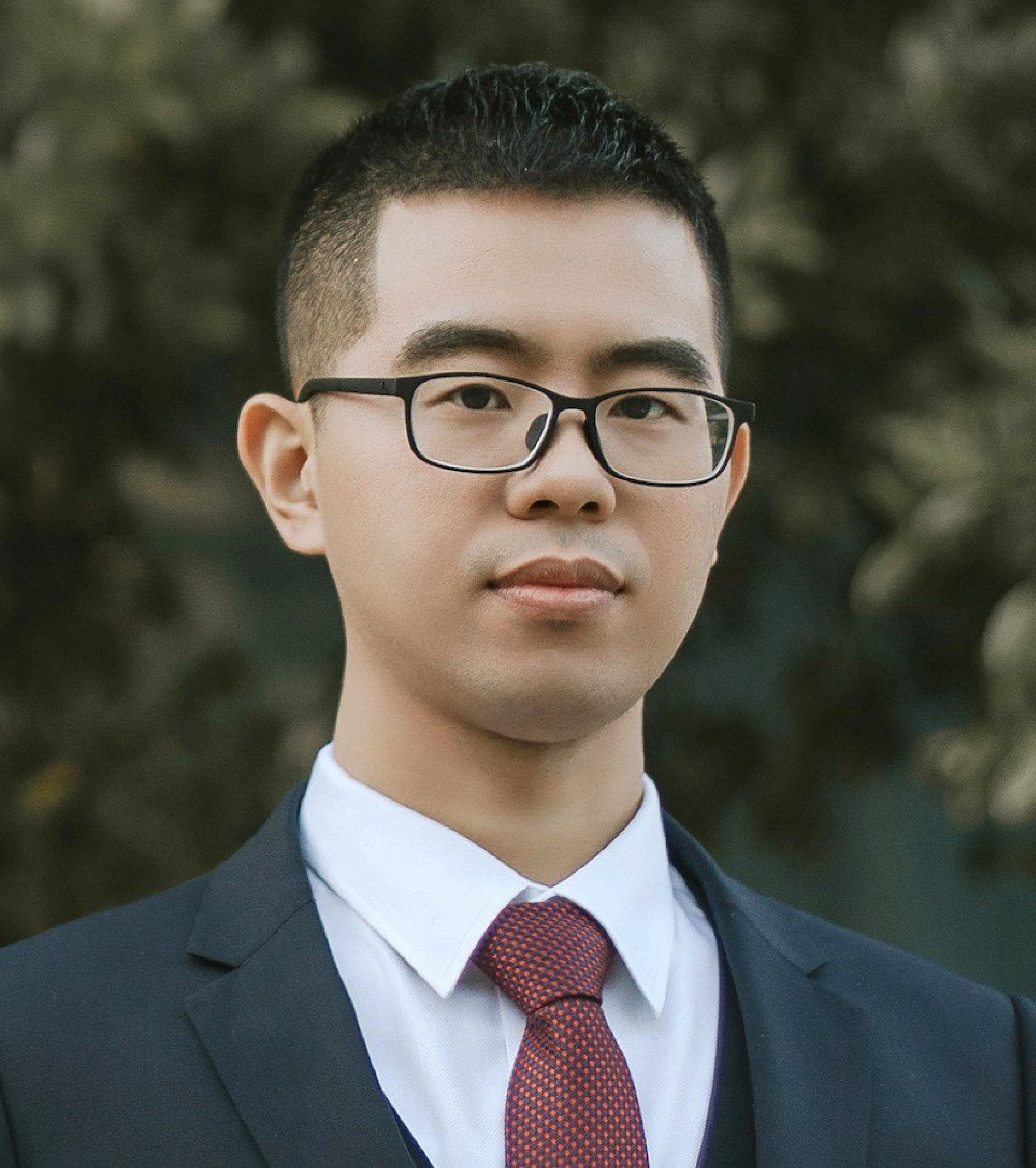}}]{Changsheng You} (Member, IEEE) received his B.Eng. degree in 2014 from University of Science and Technology of China (USTC) and Ph.D. degree in 2018 from The University of Hong Kong (HKU). He is currently an Assistant Professor at Southern University of Science and Technology, and was a Research Fellow at National University of Singapore (NUS). His research interests include near-field communications, intelligent reflecting surface, UAV communications, edge computing and learning. Dr. You is a Guest Editor for \textsc{IEEE Journal on Selected Areas in Communications} (JSAC), an editor for \textsc{IEEE Transactions on Wireless Communications} (TWC), \textsc{IEEE Communications Letters} (CL), \textsc{IEEE IEEE Transactions on Green Communications and Networking} (TGCN), and \textsc{IEEE Open Journal of the Communications Society} (OJ-COMS). He received the IEEE Communications Society Asia-Pacific Region Outstanding Paper Award in 2019, IEEE ComSoc Best Survey Paper Award in 2021, IEEE ComSoc Best Tutorial Paper Award in 2023. He is listed as the Highly Cited Chinese Researcher, Exemplary Reviewer of the IEEE Transactions on Communications (TCOM) and IEEE Transactions on Wireless Communications (TWC).
\end{IEEEbiography}

\begin{IEEEbiography}
	[{\includegraphics[width=1in, height=1.25in,clip, keepaspectratio]{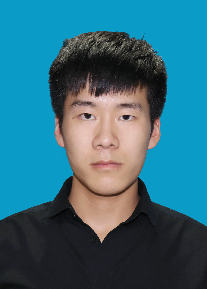}}]
	{Jiapeng Li} received the B.S. degree in information and computing science from Beijing University of Civil Engineering and Architecture, Beijing, China, in 2020, and the M.S. degree in computer science from Southwest University, Chongqing, China. He is currently pursuing the Ph.D. degree in Mathematics from Southern University of Science and Technology, Shenzhen, China. His research interests include near-field communications and lens antennas array.
\end{IEEEbiography}

\begin{IEEEbiography}
	[{\includegraphics[width=1in, height=1.25in,clip, keepaspectratio]{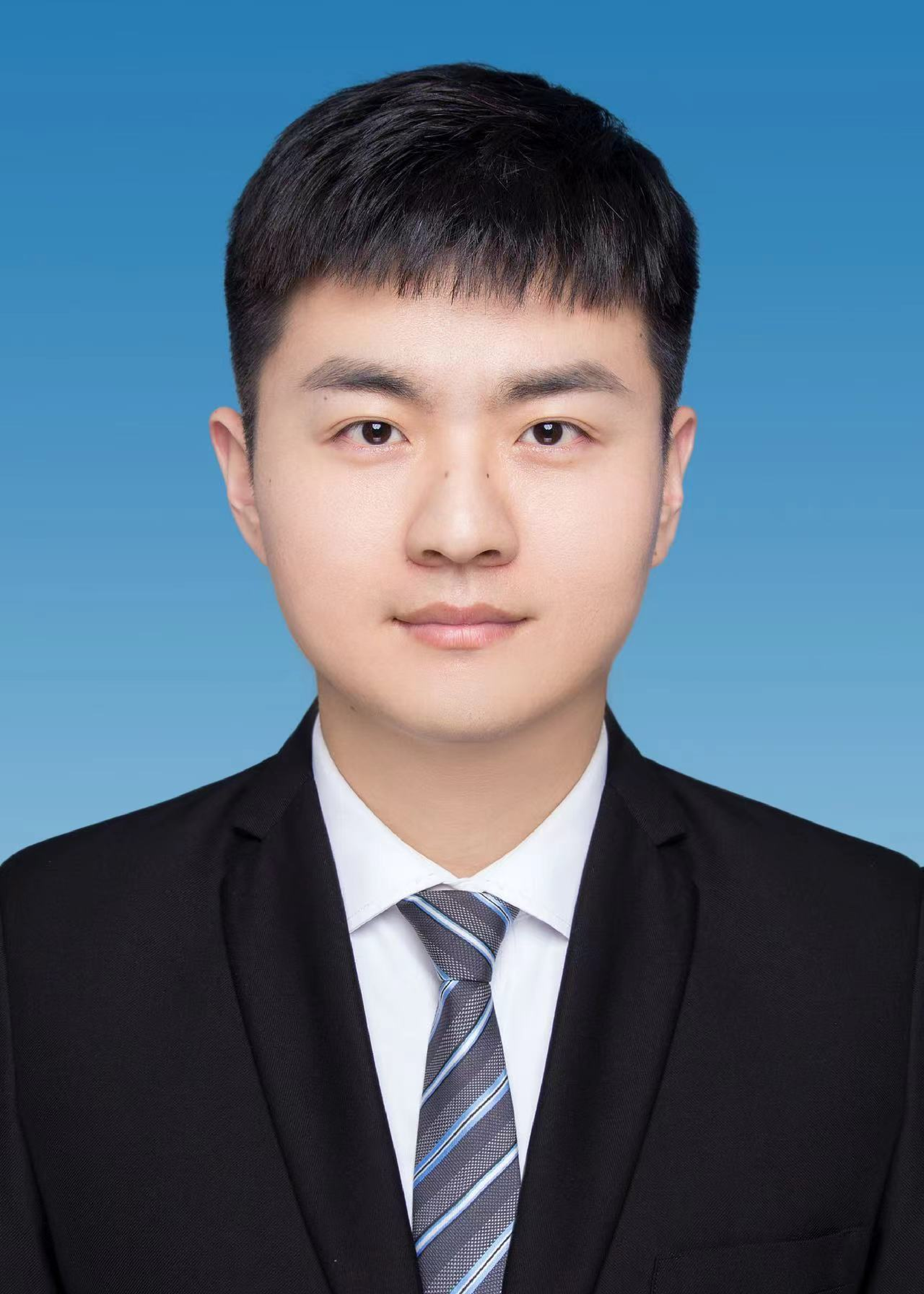}}]
{Yunpu Zhang} (Graduate Student Member, IEEE) is currently pursuing the Ph.D. degree in electrical engineering with City University of Hong Kong (CityU), Hong Kong, China. He was a Visiting Student with the Department of Electrical and Electronic Engineering, Southern University of Science and Technology, Shenzhen, China. His research interests include near-field communications  and intelligent reflecting surface.
\end{IEEEbiography}

\end{document}